\documentclass[11pt]{article}
\usepackage{amsmath,amsthm,amsfonts, amssymb}
\usepackage{multicol}
\usepackage{hyperref}
\usepackage{color}

\newtheorem{definition}{Definition}
\newtheorem{theorem}{Theorem}
\newtheorem{proposition}{Proposition}
\newtheorem{lemma}{Lemma}
\newtheorem{axiom}{Axiom}
\newtheorem{SCaxiom}{Social Choice Axiom}
\newtheorem{property}{Property}
\newtheorem{Rule}{Rule}

\newtheorem{remark}{Remark}

\def\prob#1#2{\Pr_{#1}\left[ #2 \right]}
\def\expec#1#2{\mbox{\bf E}_{#1}\left[ #2 \right]}

\def\B3CT{B$^3$CT}
\def\calT{\mathcal{T}}
\def\calC{\mathcal{C}}
\DeclareMathAlphabet\mathbfcal{OMS}{cmsy}{b}{n}
\def\calbfC{\mathbfcal{C}}
\def\calA{\mathcal{A}}
\def\calH{{\cal H}}
\def\LV{L(V)}
\def\bLV{\overline{L(V)}}

\newcommand{\overbar}[1]{\mkern 1mu\overline{\mkern-1mu#1\mkern-1mu}\mkern 1mu}

\begin{document}
\title{Fixed-Points of Social Choice: An Axiomatic Approach to Network Communities}

\author{Christian Borgs\\ Microsoft Research \and Jennifer Chayes \\ Microsoft Research
\and Adrian Marple\\ Stanford \and Shang-Hua Teng\thanks{Supported in
  part by NSF grants CCF-1111270 and CCF-0964481 and
  by a Simons Investigator Award from the Simons Foundation.}\\  USC
}

\maketitle

\begin{abstract}
We provide the first social choice theory approach to the   question of
what constitutes a community in a social network.
Inspired by social choice theory in voting and other
contexts \cite{ArrowBook}, we start
from an abstract social network framework, called
{\em preference networks} \cite{B3CT};
 these consist of a finite set of members
and a vector giving a total ranking of
the members in the set
for each of them (representing the preferences of
that member).

Within
this framework, we axiomatically study the formation
and structures of communities.
Our study naturally involves two complementary approaches.  In the
first,
we apply social choice theory
and
 define communities
 indirectly by postulating that they are
fixed points of a preference aggregation function
obeying certain desirable axioms.
In the second, we
directly postulate desirable axioms for communities without reference
to preference aggregation, leading to
a natural set of eight community
axioms.

These two approaches allow us to formulate and analyze community
rules.
We prove a taxonomy theorem that provides a
  {\em structural characterization} of the family of those
  community rules that satisfies all eight axioms.  The 
  structure is actually quite beautiful:
the family satisfying all eight
axioms forms a bounded lattice under the natural
 intersection and union operations of community rules.
The taxonomy theorem also gives an explicit characterization
of the most comprehensive community rule and the most
selective community rule consistent with all community axioms.
This structural theorem is complemented with a {\em complexity result}: we show
that while identifying a community by the selective rule is
straightforward, deciding if a subset satisfies the comprehensive rule
is coNP-complete. Our studies also shed light on the limitations of
defining community rules solely based on preference aggregation. In
particular, we show that many aggregation functions lead to
communities which violate at least one of our community axioms. These
include any aggregation function satisfying Arrow's independence
of irrelevant alternative axiom as well as commonly used aggregation
schemes like the Borda count or generalizations thereof. Finally,
we give a polynomial-time rule consistent
with seven axioms and weakly satisfying the eighth axiom.

\end{abstract}

\newpage

\section{Introduction: Formulating Preferences and Communities}

A fundamental problem in network analysis is the characterization and
  identification of subsets of nodes in a network
  that have significant structural coherence.
This problem is usually studied in the context of {\em community
  identification} and {\em network clustering.}
Like other inverse problems in machine learning, this one is
  conceptually challenging: There are many possible ways
  to measure the degree of coherence of a subset and many possible
  interpretations of affinities to model network data.
As a result, various seemingly reasonable/desirable conditions
to qualify a subset as a community have been
studied in the
  literature \cite{modularity,SpielmanTengLocalCluster,Rosvall08,mishra,B3CT,Leskovec08www,hop-hierarc,Radicchi,condkonst,KDD}.
The fact that there are an exponential number of candidate subsets
   to consider makes direct comparison of different community
   characterizations quite difficult.

Among the challenges in the study of communities in a social and
  information network are the following two basic mathematical problems:
\begin{itemize}
\item {\bf Extension of individual affinities/preferences to community coherence}:
A (social) network usually
  represents 
  pairwise interactions among its members, while the
  notion of communities is defined over its larger subsets.
Thus, to model the formation of communities, we need
  a set of consistent rules
  to extend the pairwise relations or individual preferences
  to community coherence.
\item {\bf Inference of missing links}:
Since networks typically are sparse, we also need
  methods to properly infer the missing
  links from the given network data.
\end{itemize}

In this paper, we take what we believe is a novel and principled approach to the
problem of community identification.
Inspired by the classic work in social choice theory \cite{ArrowBook},
we propose an axiomatic approach towards understanding network
communities, providing both a framework
for comparison of different community characterizations, and relating
community identification to well-studied problems in social
choice theory \cite{ArrowBook}.
Here, we focus on the
problem of defining community
rules and coherence measures from individual preferences presented in
the input social/information network, but
we think that this study will also provide the foundation for an axiomatic
approach to the
 problem of inferring missing links.
We plan to address this second problem in a subsequent
paper, which    will  use this paper as a foundation.

Through the lens of axiomatization, we examine both mathematical
and complexity-theoretic structures of communities that satisfy
a community rule or a set of community axioms.
We also study the stability of network communities, and design
  algorithms for identifying and enumerating communities with
  desirable properties.

While the study initiated here is conceptual,
we believe it will ultimately
enable a  more principled way to choose among community formation models
for interpretation of current experiments, and also suggest future experiments.

\subsection{Preference Networks}

Before presenting the highlights of our work,
  we first define an {\em abstract} social network framework
  which enables us
  to focus on the axiomatic study of community rules.
This framework is inspired by social choice theory \cite{ArrowBook}
  and was first used in \cite{B3CT} in the context
  of community identification for modeling social
  networks with complete preference information.
We will refer to each instance of this framework as
 a {\em preference network}.
Below, for a non-empty finite set $V$, let $\LV$
denote the set of all linear orders on $V$, represented, e.g., by
the set of all bijections $\pi:V\to [1:|V|]$,
 where as usual, $[n:m]$ is the set $\{n, n+1,,\dots, m\}$.
Alternatively, $\pi$ can be represented by the ordered list
$\pi=[x_1,x_2,\dots, x_{|V|}]$, where $x_i\in V$ is such that $\pi(x_i)=i$; in our
notation, $\pi(x)$ thus represents the rank of $x$ in the ordered list
$\pi=[x_1,x_2,\dots, x_{|V|}]$.

\begin{definition}[Preference Networks]
A {\em preference network} is a pair $A = (V,\Pi)$, where
  $V$ is a  non-empty finite set and $\Pi$ is a {\em preference profile on $V$},
  defined as an element $\Pi =\{\pi_i\}_{i\in V}\in \LV^V$.  Here
$\pi_i$ specifies the total ranking\footnote{
In
broader settings, one may want to consider  preferences
that allow {\em indifference} or partially ordered preferences, or
both.
One may also model a social network by a cardinal
 affinity network that specify each member's preference by a weighted
affinity vector, for example with weights from $[0,1]$ where
  $1$ and $0$, respectively, represent the highest and
  lowest preferences.
To distinguish ordinal and cardinal preferences, we refer to the
latter as an {\em affinity network}.
Both models are referred to as {\em
affinity systems} in \cite{B3CT}.
For simplicity of exposition, we first focus
  on preference networks.
In Section \ref{Sec:Remarks}, we discuss the possible
  extension of our framework.}
 of
$V$ in the order of $i$'s preference:
$\forall s,u,v\in V$, $s$ {\em prefers} $u$ to $v$,
denoted by $u\succ_{\pi_s} v$, if and only if $\pi_s(u) < \pi_s(v)$.
\end{definition}

As argued in \cite{B3CT}, a real-life
   social network may be viewed as sparse,
  observed social interactions
of an underlying latent preference network.
In this view, the communities of a preference
network may
  be considered to be the ground truth set of
  potential
communities in its observed social~network.

\subsection{Highlights of the Paper}

Our main contribution is an axiomatic framework for
   studying community formation in preference networks, and
   mathematical, complexity-theoretic, and algorithmic
   investigation of community structures in this framework.
Our work on axiomatization of network communities can be
  organized into two related parts:  (1) communities as
  fixed points of social choice aggregation functions; and (2)
  communities via direct axiomatic characterization.  In the second
  part, we specify eight axioms we would like the communities to obey,
  and find conditions under which such communities exist.  In
  the first part, we specify social choice aggregation functions
  for which the communities will be fixed points; this first method
  allows for an ``indirect'' axiomatic characterization in that
  the aggregation functions themselves could be taken to obey axioms
  \cite{ArrowBook,Young1975},
  which would then indirectly characterize the communities which
  result as fixed points.

\subsubsection*{Communities as fixed points of social choice}\label{sub:fixedpoints}

Our approach of starting from preference networks to study
communities naturally
connects community formation to social
  choice theory \cite{ArrowBook},
  which provides a theoretical framework for
  understanding the problem of combining individual
  preferences into a collective preference or decision.
In this  first part of our analysis, we use
   {\em preference aggregation functions}
    studied in
    social choice theory \cite{ArrowBook} to
    characterize communities by defining communities
   as fixed points of a preference aggregation
    function.

Since real-world voting schemes and preference aggregation functions
do not always produce a total order, we will use the following
notation in the definition below.  Let $\bLV$ denote the set of
all {\em ordered partitions} of
 $V$.
For a
$\sigma \in \bLV$, for $i,j\in V$, we use
$i\succ_{\sigma} j$ to denote that $i$
is {\em strictly preferred to} $j$ (that is, $i$ and $j$ belong to
different partitions, and the partition containing $i$ is ahead of the
partition containing $j$ in $\sigma$). In this case, we also say
$j \prec_{\sigma} i$.  We will use $i\succeq_{\sigma} j$ to denote
that $i \succ_{\sigma} j$ or $i$ and $j$ are in the same partition.

To continue, we need some notions motivated by social choice theory.
In this context, $V$ will be considered a set of ``candidates''.  We'll also need a set of possible voters, $\cal S$, which is assumed to be
a countable set -- if not otherwise specified, we identify $\cal S$ with the positive integers $\mathbb N$.  With a slight abuse of notation,
we denote the union of $\LV^S$ over all non-empty finite $S\subseteq
  \mathcal S$ by
$\LV^*$.  A {\em preference aggregation function} is then defined to be an
arbitrary function
 $F: \LV^* \rightarrow \bLV$.  Given a non-empty finite set of voters $S$ and a preference profile
$\Pi_S = \left\{\pi_s \ : \ s\in S\right\}\in \LV^S$, the image $F(\Pi_S)$ is called the
{\em aggregated preference}
\footnote{Note that in our notation without further requiring $F$
to satisfy additional conditions, the labels in $S$ matter: e.g., even if $\pi_1$ and $\pi_2$ are the same permutation of $[n]$, the values of $F(\Pi_{\{1\}})$ and $F(\Pi_{\{2\}})$ can be different.}
of $S$.

\begin{definition}{\sc (Communities as Fixed Points of Social Choice)}\label{SocialChoiceVoting}
Let $A = (V,\Pi)$ be a 
preference network,
$F: \LV^* \rightarrow \bLV$ be a preference aggregation function,
and $\emptyset\neq S\subseteq V$.  $S$ is called a
{\em community of $A$ with respect to $F$} if and
only if $u \succ_{ F(\Pi_S)} v$, $\forall u\in S, v\in V-S$.

The function $\calC_F$ mapping $A$ into the set
  of communities defined above is called
  the {\em fixed point rule with respect to $F$}.
If $F$ is not specified, i.e., if there exists an $F$ such that $\calC=\calC_F$, we call $\calC$ simply {\em a fixed point rule}.
\end{definition}

Informally,  this definition says  that
  a community
is a subset  $S\subseteq V$
such that, when we aggregate the preferences of all
its members, 
the resulting
aggregated preference puts
the members $S$ as the top $|S|$ elements.\footnote{In the case of ties, we allow for ties
among the top $|S|$ members, as well as among the lower ranked members, but not
between the top $|S|$ members and anyone below.} In other words, under
the aggregation function $F$, the members of the
community ``vote'' for themselves.
Thus, $S$ is a {\em fixed point} of its aggregated preference.
The community characterization of Definition \ref{SocialChoiceVoting}
   generalizes the following concept of self-determination of
   \cite{B3CT}:

\begin{definition}{\sc (\B3CT Communities)}\label{B3CTVoting}
Let $A = (V,\Pi)$ be a preference network.
For  $\emptyset\neq S\subseteq V$ and $i\in V$, let $\phi_S^\Pi(i)$
denote the number of votes that member $i$ would receive if each
member $s\in S$ was casting one vote for each of its $|S|$ most
preferred members according to its preference $\pi_s$.
In other words,
$\phi_S^{\Pi}(i)
=\left|\{ s\ :\ (s\in S) \ \&\
 (\pi_s(i)\in [1:|S|])\}\right|.$
Then, $S$ is {\em \B3CT-self-determined} if
  everyone in $S$ receives more votes from
  $S$ then everyone outside $S$.
\end{definition}

It is easy to see that the \B3CT voting rule is an instance of
a fixed-point rule, with preference aggregation function $F$
defined by
$v\succ_{F(\Pi_S)} w$ iff $\phi_S^\Pi(v)>\phi_S^\Pi(w)$.

We will also refer to a community according to
   Definition \ref{SocialChoiceVoting} as an {\em $F$-self-determined community}.
We are particularly interested in those aggregation functions
   that satisfy various axioms in social choice theory \cite{ArrowBook},
since this
  enables us to  utilize  established social choice
  theory
  to   study all conceivable
  self-determination community rules
   within one unified framework.
For example, it allows us to
  reduce the fairness analysis for community formation
  to the fairness of preference aggregation functions.

Arrow's celebrated impossiblility theorem and subsequent work
  in social choice theory \cite{ArrowBook}
  point to both challenges and exciting opportunities for
  understanding communities in preference networks.
Recall that Arrow's theorem states that for $n>2$, no
  (strictly linear)
preference aggregation function satisfies all of the following
  three  axiomatic conditions: {\sf Unanimity},
  {\sf Independence of Irrelevant Alternatives}, and {\sf
  Non-Dictatorship}
(see Section \ref{Sec:Axioms} for
definitions.)
On the other hand, preference aggregation functions
  exist if one
  relaxes any
  of these
  conditions.
For instance, the well-known Borda count \cite{YoungBorda} is a
unanimous
  voting method
  with no dictators.

In this paper, we will examine the impact of preference aggregation
  functions
  on
  the structure of the self-determined communities that they
  define, as well as the limitations of
formulating community rules solely based on preference aggregation.
See below for more discussion.

\subsubsection*{Communities via direct axiomatic characterization}\label{sub:axioms}

In this second approach, we will use a more direct axiomatic characterization to
  study network communities.
To this end, we use a set-theoretical {\em community function} as a means
  to characterize a community rule.

\begin{definition}[Community Functions] \label{defi:community}
Let $\calA$ denote the set of all preference networks.
A {\em community function}
is a function $\calC$ that maps a preference network $A = (V,\Pi)$
  to a characteristic function of non-empty subsets of $V$.
In other words,  $\calC(A) \subseteq 2^{2^V{-\{\emptyset\}}}$
  is an indicator function of $2^V{-\{\emptyset\}}$.
We say a subset $S\subseteq V$ is a {\em community} in a
  preference network $A = (V,\Pi)$ according to
  a community function
  $\calC$  if and only if $S\in \calC(A)$.
To simplify our notation, for $A=(V,\Pi)$ we often write $\calC(V,\Pi)$ instead
  of $\calC((V,\Pi))$.
\end{definition}
We use axioms to state
properties, such as fairness and
  consistency, that a desirable community
  function should have when applied to all preference networks.
An example is the property that the community function
should be isomorphism-invariant:
Here an {\em isomorphism} between
two preference networks $A=(V,\Pi)$ and $A' = (V',\Pi')$ is
a bijection $\sigma: V\rightarrow V'$ such that
$\Pi'=\sigma(\Pi)$, i.e., such that
for all $s,v\in V$, $\pi'_{\sigma(s)}(\sigma(v)) =\pi_{s}(v)$,
and two  preference networks $A$ and $A'$ are
{\em isomorphic} to each other if there exists
such an isomorphism.
Isomorphism invariance then requires that
for any pair of isomorphic preference networks $A = (V,\Pi)$
and $A'=(V',\Pi')$ and any
isomorphism $\sigma$ between $A$ and $A'$,
if $S \subset V$ is a community in $A$, then $\sigma(S)$ should still
be a community in the $A'$.
Another example is the property of
 {\em monotonic characterization}:
  If $S$ is a community in $A = (V,\Pi)$,
  then $S$ should remain a
  community in every preference network
  $A' = (V,\Pi')$
such that
for all $u,s\in S$ and $v\in V$,
  if  $u \succ_{\pi_s}v$ then $u \succ_{\pi'_s}v$.

In Section \ref{Sec:Axioms}, we propose a natural set of eight
  desirable community axioms.
Six of them, including both examples above, provide a positive
  characterization of communities.
These axioms concern the consistency, fairness, and robustness of a
  community function, as well as the community structures when a
  preference network is embedded in a larger preference network.
The other two axioms address the necessary stability and self-approval
  conditions that a community should satisfy.

\subsubsection*{Constructing and Analyzing Community Rules}

While Definition \ref{defi:community}
is convenient for the study of the mathematical structure of our theory,
community identification is a computational problem as
  much as a mathematical problem.
Thus, it is desirable that communities
  can be characterized by a constructive community function $\calC$
  that is:
\begin{itemize}
\item {\bf Consistent}: $\calC$ satisfies all (or nearly all) axioms;
\item {\bf Constructive}: Given a preference network
  $A = (V,\Pi)$, and a subset $S \subseteq V$, one can determine in polynomial-time
  (in $n=|V|$) if
   $S \in \calC(A)$.
\item {\bf Samplable}: One can efficiently
obtain a random sample of $\calC(A)$.
\item {\bf Enumerable}: One can efficiently enumerate $\calC(A)$, for instance,
  in time $O(n^k \cdot |\calC(A)|)$ for a constant $k$.
\end{itemize}
Our two axiomatic approaches allow us to formulate a rich family of community
  rules
and analyze their properties.
Using the fixed-point rule,
  we can define a constructive community function based on any
  polynomial-time
  computable aggregation function.
Alternatively, we can use one axiom or
a set of
  axioms as a community rule.
We can also define a community rule by the intersection
  of a fixed-point rule and a set of axioms.
In this paper, we aim to {\em characterize the community rules that satisfy
 a set of ``reasonable'' axioms},
  and address the basic questions:
\begin{quote}
\begin{itemize}
\item
{\em Is there an aggregation function leading to a community
  rule satisfying this set of ``reasonable'' axioms?}
\item {\em What is the complexity of the community rules based on these axioms?}
\item {\em How are different community rules satisfying our axioms related to each other? For example, given two community rules $\calC_1$ and $\calC_2$ satisfying our axioms, does the rule $\calC$ defined by
    $\calC(A) :=\calC_1(A)\cap\calC_2(A)$ obey our axioms as well?}
\end{itemize}
\end{quote}

\subsubsection*{Structural and Complexity-Theoretic Results}

Our main structural result is a taxonomy theorem that provides a complete
  characterization of the most comprehensive community rule and the most
  selective community rule consistent with all our community axioms.
This result illustrates an interesting contrast to
  the classic axiomatization result of Arrow \cite{ArrowBook}
  and the more recent result of Kleinberg on
  clustering \cite{KleinbergClustering}
 that inspired our work.
Unlike voting or clustering where the basic axioms lead to
  impossiblity theorems,
  the preference network framework
  offers a natural community rule,
  which we call the {\em Clique Rule,}
  that is intuitively
  fair, consistent, and stable, although selective (See Section
\ref{sec:taxonomy} for more details):
$S$ is a community according the Clique Rule
  iff each member of $S$ prefers every member of $S$ over every non-member.
Indeed the Clique Rule satisfies all our axioms.
Our  analysis then leads us to a community rule which is consistent with all axioms -- we
 call it the {\em Comprehensive Rule} --
 such that for any community rule $\calC$
  satisfying all axioms and all preference network $A$,
 $\calC_{clique}(A) \subseteq \calC(A) \subseteq \calC_{comprehensive}(A)$.
Perhaps more interesting, under the natural operations of union and intersections,
 the set of all community rules satisfying all our axioms becomes a lattice with
 $\calC_{clique}(A)$ and $\calC_{comprehensive}(A)$ forming a lower and upper bound, respectively.

We complement this structural theorem with a complexity result: we show
that while identifying a community by the Clique Rule is
straightforward, it is {\sf coNP-complete} to determine if a subset
satisfies the comprehensive rule.

Our studies also shed light on the limitations of
  formulating community rules solely based on
  preference aggregation.
In particular, we show that many aggregation functions lead to
  communities which violate at least one of our community axioms.
We give two impossibility-like theorems.
\begin{enumerate}
\item  Any fixed-point rule based on commonly used aggregation
  schemes like Borda count or generalizations thereof -- such as the
 \B3CT self-determination rule -- is
  inconsistent with (at least) one of our axioms.
\item  For any aggregation function satisfying Arrow's independence
 of irrelevant alternative axiom, its fixed-point rule
 must violate one of our axioms.
\end{enumerate}
Finally, using our direct axiomatic framework, we analyze the following
  natural constructive community function inspired by preference
  aggregations.

\begin{definition}[Harmonious Communities]\label{def:harmonious}
A non-empty subset $S \subseteq V$ is a {\em harmonious community} of
  a preference network $A  = (V,\Pi)$ if for all $u\in S$ and $v \in V
  -S$,
  the majority of    $\{\pi_s\ : \ s\in S\}$  prefer $u$ over $v$.
\end{definition}

We will show that the harmonious community rule
  is consistent with seven axioms and  satisfies a
  weaker form of the eighth axiom.
In addition, various stable versions of harmonious communities (see
  the discussion below) enjoy some degree of samplablility
  and enumerability.

\subsubsection*{Stability of Communities and Algorithms}\label{sec:stability}

In real-world social interactions, some communities
  are more stable or durable than others when people's
  interests and preferences evolve over time.
For example, some music bands stay together longer than others.
Inspired by the work of \cite{B3CT} and Mishra {\em et al.}
\cite{mishra} on modeling this phenomenon,
  we examine the impact of stability
  on the community structure.

To motivate our discussion, we first recall the main definition and
result of \cite{B3CT}:
\begin{definition}
\label{def:a-b-community}
For $0 \leq \beta < \alpha \leq 1$, a non-empty  subset $S \subseteq V$ is an
{\em $(\alpha,\beta)-$\B3CT community} in $A = (V,\Pi)$
iff $\phi_S^\Pi(u) \geq \alpha\cdot |S|$ $\forall u\in S$ and
  $\phi_S^\Pi(v) < \beta\cdot |S|$ $\forall v\not\in S$.
\end{definition}
It was shown in \cite{B3CT} that, in any preference network, there are only
polynomially many stable \B3CT communities when the parameters
$\alpha, \beta$ are constants, and they can be enumerated in
polynomial time, showing
 that the strength of
community coherence has both structural and computational
implications.

In Section \ref{Sec:Stability}, we consider several stability
  conditions in our axiomatic community framework.
In one direction, we examine the structure of the communities
   (defined by a fixed-point community rule) that remain self-determined even after
  a certain degree of perturbation in its members' preferences.
In this context, for example, we can reinterpret the \B3CT-stability
defined above
as follows: A subset $S \subseteq V$ is an
  {\em $(\alpha,\beta)-$\B3CT community} in a preference network $A$
   if it remains self-determined when $|S|\cdot
  (\alpha-\beta)/2$ members of $S$ make arbitrary changes to their preferences.
In the other direction, we consider some notions of
 stability derived directly from the
 social-choice based community framework
 where members of a community
  separate themselves from the rest.
We can further use the separability as a
  measure of the community strength and stability
  to capture the intuition that
  stronger communities are also themselves more integrated.
As a concrete example, we show in Section \ref{Sec:Stability} that there are
 a quasi-polynomial number of stable harmonious communities for
 all these notions of  stability.
This result demonstrates that there exists a constructive community function that
  essentially satisfies all our axioms, whose stable
  communities are quasi-polynomial-time samplable and enumerable.

\section{Coherent Communities: Axioms}\label{Sec:Axioms}

In this section, we define our eight core axioms, give a more formal treatment of social choice axioms, and examine several properties of community rules and the relations these have with each other.

\subsection{Lexicographic Preference}
\label{sec:lex-pref}

The following notion will be crucial in several parts of this paper, and is implicitly used in our first two axioms below.

\begin{definition}[Lexicographical Preferences] 
Given a preference network $(V,\Pi)$ and two non-empty disjoint
subsets $G$ and $G'$ of equal size, we say that $s\in V$ {\em lexicographically prefers} $G'$ over $G$  if there exists a bijection $f_s: G \rightarrow G'$ such that $ f_s(u)\succ_{\pi_s}u$ for all $u\in G$.

We say that a group $T\subset V$ lexicographically prefers $G'$ over $G$ if every member $s\in T$ lexicographically prefers $G'$ to $G$, i.e., if there exists a set of bijections $\{f_s: G \rightarrow G'\mid s\in T\}$ such that $ f_s(u)\succ_{\pi_s}u$ for all $u\in G$ and all $s\in T$.
\end{definition}

Note that, in contrast to the standard lexicographical order, lexicographical preference is only a partial order.  The notion is motivated by the following proposition.

\begin{proposition}\label{prop:lex}
	Let $\pi\in\LV$, let $G$ and $G'$ be  disjoint subsets of $V$ with $|G| = |G'|$. Let  $G[i]$
(and $G'[i]$)
be the $i^{th}$ highest ranked element of $G$
(and $G'$) according to $\pi$.
Then
 there exists a bijection $f: G\longrightarrow G'$ such that for all $g\in G$, $f(g) \succ_{\pi} g$ if and only for all $i\in [1:|G|]$,  $G'[i] \succ_{\pi} G[i]$.
\end{proposition}
\begin{proof}
Suppose $f$ satisfies the condition of the proposition.  Then $G'[1]\succeq_{\pi} f(G[1]) \succ_{\pi} G[1]$.
If
$G'[1]\neq f(G[1])$, define $h$ to be the bijection on $G'$ which exchanges $G'[1]$ and
$ f(G[1])$, and define $\tilde f=g\circ f$.  Then $G'[1]=\tilde f(G[1]) \succ_{\pi} G[1]$
while $\tilde f$ still satisfies the condition of the proposition.
Removing $G[1]$ from $G$ and $G'[1]$ from $G'$, we continue by induction to prove the only if statement.  The if statement is obvious - just define $f$ by $f(G[i])=G'[i]$.
\end{proof}

\subsection{Axioms for Community Functions}\label{sub:CommunityAxioms}

For the following definitions, fix a ground set $V$ and a community function $\calC$.

\begin{axiom}[{\sf Group Stability (GS)}]
If $\Pi$ is a preference profile over $V$ and $S\in\calC(V,\Pi)$, then
$S$ is {\sc group stable} with respect to $\Pi$.  Here a subset
$S \subset V$ is called {\em group stable with respect to} $\Pi$ if
for all  non-empty  $G \subsetneq S$,  all $G' \subset V - S$ of the same
size as $G$, and all tuples of bijections,
  $(f_i: G \rightarrow G', i\in S-G)$,
  there exists $s \in S - G$, $u\in G$ such that $u \succ_{\pi_s} f_s(u)$.
\end{axiom}

This axiom provides a type of game-theoretic stability
\cite{NashNonCooperative,NAS50,core1,core2,core3},
  and states that no subgroup in a community can be replaced by an equal-size
group of non-members that are lexicographically  preferred by the
remainder of the community members.
For instance, if the subgroup is
of size 1, this means that there is no outsider that is universally
preferred to this member, excluding that member's own opinion.  On the
other end of the spectrum, if the subgroup is all but one person, then
group stability states that there must be someone from that member's
top choices, and thus represents a type of individual rationality
condition. Note that 
the set $V$ is vacuously
group stable for all $\Pi$.

\begin{axiom}[{\sf Self-Approval (SA)}]
If $\Pi$ is a preference profile over $V$, and $S\in\calC(V,\Pi)$ then $S$
is {\sc self-approving} with respect to $\Pi$.
Here a subset $S \subset V$ is called {\em self-approving with respect to}
$\Pi$ if for all 
 $G' \subseteq V - S$ of the same size as $S$,
and all tuples of bijections $(f_i: S \rightarrow G', i\in S)$
 there exists $s, u \in S$, such that $u \succ_{\pi_s} f_s(u)$.
\end{axiom}

Axiom {\sf SA} uses the same partial ordering of groups as the first, and  requires that there is no outside group of the same size as $S$ which is lexicographically preferred to $S$ by  everyone in $S$.
It generalizes the intuition that a singleton should be a
community only if that member prefers herself to everyone else.  Note that 
any set $S$ of size larger than $|V|/2$ 
is vacuously self-approving for all $\Pi$.

\begin{axiom}[{\sf Anonymity (A)}]
Let $S, S'\subset V$ and $\Pi,\Pi'$ be such
$S'=\sigma(S)$ and $\Pi'=\sigma(\Pi)$ for some
permutation $\sigma: V \rightarrow V$.
Then	
$S \in \calC(V,\Pi)
        \Longleftrightarrow S' \in \calC(V,\Pi')$.
\end{axiom}

A staple axiom, {\sf Anonymity}, states that labels should have
  no effect on a community function.

\begin{axiom}[{\sf Monotonicity (Mon)}]
	Let $S \subset V$.  If $\Pi$ and $\Pi'$ are such that for all $s\in S$
$$
u \succ_{\pi_s'}v\Longrightarrow u\succ_{\pi_{s}} v \quad\text{ for all }u\in S,v\in V
$$
then $S \in \mathcal{C}(V,\Pi')\Longrightarrow S \in \mathcal{C}(V,\Pi).$
\end{axiom}

The Axiom {\sf Monotonicity} states that, if a member of a community gets promoted
without negatively impacting other members, then that subset must
remain a community.  Thus this axiom reflects the fact that high
positions imply greater affinities towards those people.
Note that
  {\sf Mon} also allows non-members to change arbitrarily, as long as
 their positions relative to any members remains the same or worse.

\begin{axiom}[{\sf Coherence Robustness of Non-Members (CRNM)}]
Let $S \subset V$.  If $\Pi$ and $\Pi'$ are such that for all $s , t\in S$
$$
v \succ_{\pi_s'} w \Longleftrightarrow v \succ_{\pi_t'} w
\quad\text{for all }v,w\notin S
$$
and
$$
\pi_s'(u) = \pi_s(u) \quad\text{  for all }u\in S,
$$
then $S\in \mathcal{C}(V,\Pi')\Longrightarrow S\in \mathcal{C}(V,\Pi)$.
\end{axiom}

\begin{axiom}[{\sf Coherence Robustness of Members (CRM)}]
Let $S \subset V$.  If $\Pi$ and $\Pi'$ are such that for all $s , t\in S$
we have
$$
u \succ_{\pi_s'} w\Longleftrightarrow u \succ_{\pi_{t}'} w
\quad\text{for all }u, w \in S
$$
and
$$
\pi'_{s}(v) = \pi_{s}(v)\quad\text{for all }v \notin S,
$$
then $S\in \mathcal{C}(V,\Pi')\Longrightarrow S\in \mathcal{C}(V,\Pi)$.
\end{axiom}
The two {\sf Coherence Robustness} Axioms reflect the fact that, if
community members agree about their preferences concerning either members
or non-members, they are less likely to be a community.
In the
  case of non-members, agreement implies that some non-member is more
  preferred and therefore more likely to break up the community.
  Contrariwise, in the case of members, agreement implies some member
  is less preferred and more likely to be ousted.
\begin{axiom}[{\sf World Community (WC)}]
For all preference profiles $\Pi$,
	$V \in \calC(V,\Pi)$.
\end{axiom}

To state the next axiom, we define the projection $\left.A\right|_{V'}$ of a preference
network $A=(V,\Pi)$ onto a subset $V'\subset V$ as the preference network
$\left.A\right|_{V'}=(V',\left.\Pi\right|_{V'})$ where
$\left.\Pi\right|_{V'}=\{\pi'_s\}_{s\in V'}$ is defined by setting
$\pi_s'$ to be the linear order on $L(V')$ which keeps
the relative ordering of all members of $V'$, i.e., for all $s,u,v\in V'$,
$u\succ_{\pi'_s} v\Longleftrightarrow u\succ_{\pi_s} v$.
We say that $A'$ is {\sc embedded} into $A$ if $A'=\left. A\right|_{V'}$ for some
$V'\subset V$.

\begin{axiom}[{\sf Embedding (Emb)}]
If $A'=(V',\Pi')$ is embedded into $A=(V,\Pi)$ and
$\pi_i(j) = \pi'_i(j)$ for all $i,j \in V'$ then
$${\cal C}(A') = {\cal C}(A)\ \cap\ 2^{V'}.$$
\end{axiom}

In other words, if  a  network $(V',\Pi')$ is embedded into a larger network $(V,\Pi)$
in such a way that,
with respect to the preferences in the larger network, the members of the smaller network
prefer
each other over everyone else, then the set of communities in the
larger network which are subsets of $V'$ is identical to the  set of communities
in the smaller network.

Note that, in contrast to the first seven axioms, which refer to a fixed finite ground set $V$, the last axiom
links different grounds sets to each other.
Strictly speaking, a community rule $\calC$ is therefore not just one function
$\calC: (V,\Pi) \longmapsto 2^{2^{V}{-\{\emptyset\}}}$, but a collection of such functions, one for each finite set $V$
contained in some countable reference set, say the natural numbers
\footnote{While we use the embedding axiom to makes statements about subsets of a given ground set $V$,
see, e.g., Propositions~\ref{prop:Clique-Null} and \ref{prop:MonEb->OD}
 below, we never use that we can embed a given preference network into an even larger one.
Therefore, all results of this paper, except for those involving complexity statements, hold if one
restricts oneself to a finite set $V_0$, and only considers preference networks defined on subsets
$V\subset V_0$.}
$\mathbb N$.
In a similar way, preference aggregation is not defined by a single function $F:L^*(V)\to \overline{ L(V)}$ but by a
set of such functions, one for each finite $V$ contained in the reference set.  However, when we define preference aggregation,
we usually define it for a fixed $V$, leaving the dependence on $V$ implicit.

Note also that together, Axioms {\sf Anonymity} and {\sf Embedding} imply the isomorphism invariance discussed in
the introduction.

\subsection{Properties of Social Choice Axioms}\label{sub:SCaxioms}

Before we begin to study the properties induced by social choice
axioms, we look at the properties that fixed point rules have without
any further assumptions.  To this end, we will define two properties
of a community rule $\calC$.
\begin{property}[{\sf Independence of Outside Opinions (IOO)}]
A community funcion $\calC$
  satisfies {\sf Independence of Outside Opinions}
  if, for all subsets $S\subseteq V$ and all pairs of preference profiles $\Pi, \Pi'$
  on $V$
  such that $\pi_s' = \pi_s$ for all $s\in S$, we have that
$$S \in \calC(V, \Pi') \Longleftrightarrow S\in \calC(V, \Pi).$$
\end{property}

Property {\sf IOO} simply states that the preferences of outsiders
cannot influence whether or not a subset is a community.
It turns out that 
this property (and one of our Axioms) is 
  always satisfied by any fixed-point community rule.

\begin{proposition}
\label{prop:fix-point-props}
All fixed-point rules satisfy {\sf Independence of Outside Opinions}
and {\sf World Community}.
\end{proposition}
\begin{proof}
	Clearly, any fixed point rule satisfies {\sf IOO} since the preferences of outsiders are entirely ignored when deciding if a subset constitutes a fixed point.
The axiom {\sf WC} is satisfied vacuously, because 
 it involves looking at all 
$v \in V - V$.
\end{proof}

Turning now to social choice axioms, we must first formally define the axioms informally described in Section~\ref{sub:fixedpoints}.
To this end,  we need the notion of an election, which will be defined as a triple $(V,F,S)$
where $V$ and $S$ are finite sets (called the set of candidates and voters, respectively), and $F:\LV^*\to \bLV$ is a preference aggregation function.

\begin{SCaxiom}[{\sf Unanimity (U)}]
\label{axiom:SCpareto}
 An election $(V,F,S)$ satisfies
\emph{Unanimity} if, for all
preference profiles, $\Pi_S=\{\pi_s\colon s\in S\}\in \LV^S$ and all pairs of candidates, $\{i,j\} \subseteq V$,
$$
 \pi_s(i) > \pi_s(j), \forall s\in S \Longrightarrow F(\Pi_S)(i) > F(\Pi_S)(j).
$$
\end{SCaxiom}

The question then is: what properties capture the intuition behind {\sf Unanimity}
and how do they relate to this social choice axiom?
To  answer this, we define the following two properties of a community function $\calC$.

\begin{property}[{\sf Pareto Efficiency (PE)}]
\label{axiom:PE}
	A community function, $\calC$, is {\sf Pareto Efficient}
        if, for a given preference network $A$
        and a given community $S \in \calC(A)$,
  it is the case that for all $u \in S$, $v \notin S$,
   there is a $s \in S$ such that $u \succ_{\pi_s} v$.
\end{property}

\begin{property}[{\sf Clique (Cq)}]
\label{axiom:Cq}
	A community function $\calC$
   satisfies the {\sf Clique} Property  if
for all  $A = (V,\Pi)$,
\[ u \succ_{\pi_s} v, \forall u,s \in S, \forall v \notin S \Longrightarrow S \in \mathcal{C}(A).\]
\end{property}

Property {\sf Pareto Efficiency} is a negative property that states that subsets
in which a non-member is  preferred to a member by everyone inside the subset,
 should not be a community.
 In contrast, {\sf
Clique} is a positive Property, in that it states that a completely self-loving group (i.e., a clique)
must be a community.

It turns out that both of these properties are implied by {\sf Unanimity}.

\begin{proposition}
\label{prop:Pareto}
	Fix $V$ and a preference aggregation function $F$, and let $\calC_F$ be the fixed point rule with respect to $F$.
If all elections $(V,F,S)$ with $S\subsetneq V$ satisfies {\sf Unanimity},
  then $\calC_F$ satisfies
  the properties {\sf Pareto Efficiency} and {\sf Clique}.
\end{proposition}
\begin{proof}
	Fix a preference network $A = (V,\Pi)$.

	First, let us show that $\calC_F$ satisfies {\sf Pareto
          Efficiency}.  Assume otherwise.  In this case there must be
        a community $S\subsetneq V$ such that for some $s \in S$ and
        $j \notin S$, everyone in $S$ prefers $j$ to $s$.  However,
        this implies that $j$ must be ranked higher than $s$ in
        $F(\Pi_S)$ by {\sf Unanimity}.  By the pigeon hole principle
        this implies that the elements of $S$ cannot occupy the first
        $|S|$ positions of this preference aggregation, and therefore
        $S$ is not a community.
	
	Now to show that $\calC_F$ satisfies the {\sf Clique}
        Property,  assume $S\subsetneq V$ is a clique
        ($\forall i,j \in S$ and $k \notin S$, $j \succ_{\pi_i} k$).
        Then all elements of $S$ are preferred by all members of
        $S$ to all members of $V-S$ and therefore must appear in the
        first $|S|$ slots of $F(\Pi_S)$ by {\sf Unanimity}.  This then implies that $S$ is a community as required.
\end{proof}

\begin{SCaxiom}[{\sf Non-Dictatorship (ND)}]\label{axiom:SCND}
 An election $(V,S,F)$ is {\sf Non-Dictatorial} if there exists no
 dictator, i.e., no voter $i\in S$
such that $F(\Pi_S)=\pi_i$ for all preference profiles $\Pi_S\in \LV^S$.
\end{SCaxiom}
Instead of showing properties
    implied by {\sf ND} as we did with {\sf Unanimity}, we do the
    inverse, and show that a dictatorship violates some of our
    axioms.

\begin{proposition}
\label{prop:ND}
Fix $V$ and a preference aggregation function $F$.
 If $\calC_F$, the fixed point rule with respect to $F$, satisfies {\sf Group Stability} or {\sf Anonymity}, then all elections $(V,F,S)$ with $S\subset V$ and $1<|S|<|V|$ satisfy {\sf Non-Dictatorship}.
\end{proposition}
\begin{proof}
Assume $(V,F,S)$ is dictatorial, with dictator $s \in S$.
Let $\pi_s$
be such that all members of $S$ are ranked above those outside of $S$.
Because $s$ is a dictator, we have that $S$ is a community ($S \in
\calC_F$).  Additionally let every other member of $S$ rank some
non-member $v \notin S$ above $s$.
	
	However, if $\calC_F$ satisfies {\sf Group Stability}, $S$ cannot be a community.  Furthermore, if $\calC_F$ satisfies {\sf Anonymity}, if the preferences of any two members of $S$ are swapped, $S$ should remain a community.  However, if  $s$ swaps with any other member of $S$, $v$ will be ranked above $s$ in the aggregate preference and thus $S$ cannot be a community.
\end{proof}

The last of the three social choice axioms, {\sf Independence of Irrelevant Alternatives}, simply states that the aggregate relation between any two pairs of candidates should not depend on the preferences for any other candidate.

\begin{SCaxiom}{\sc ({\sf Independence of Irrelevant Alternatives})}
\label{axiom:SCIOO}
	An  election $(V,F,S)$ satisfies {\sf Independence of Irrelevant Alternatives (IIA)} if for
all preference profiles, $\Pi_S, \Pi'_S \in \LV^{S}$ { and all candidates $a,b\in V$ we have that}
\[\left(\forall s\in {S}, a \succ_{\pi_s} b
\Leftrightarrow a \succ_{\pi'_s} b\right)
\Longrightarrow \left(a \succ_{F(\Pi_{S})} b
\Leftrightarrow a \succ_{F(\Pi'_{S})} b\right).\]
\end{SCaxiom}

This axiom is can reasonably be considered the strongest of the three,
in that it says that the aggregate preference between two candidates does not even depend on the preferences voters have between either of the two and some other candidate.  We will demonstrate this strength by proving an impossibility result involving modest assumptions about the fixed point rule of an aggregation function that satisfies {\sf IIA}.

\begin{theorem}
\label{thm:impossibility}
{ Let $F$ be an aggregation function such that the fixed point rule with respect to
$F$ satisfies the {\sf Clique} Property and the {\sf Group Stability}
Axiom.  Then no election $(V,F,S)$ with $S\subseteq V$ and $1<|S|<|V|$
satisfies
{\sf IIA}.}
\end{theorem}

\begin{proof}
Let $S\subseteq V$ such
that $1 < |S| < |V|$.
Assume that the election $(V,F,S)$
satisfies {\sf IIA}, and the resulting fixed point
rule $\calC_F$ satisfies {\sf Cq} and {\sf GS}.
We will first show that the election $(V,F,S)$ must satisfy
{\sf Unanimity}.

 In the
following preference profiles, $\Pi$, $\Pi'$, $\Pi''\in \LV^S$, we
assume that every member of $S$ has the same preference,
$\pi$, $\pi'$, and $\pi''$ respectively.  First, let $\pi$
rank all members of $S$ above non-members.  By the
{\sf Clique} Property, $S \in \calC_F(A)$ and thus
\begin{equation}
\label{partial-unam1}
\forall s \in S, v \notin S, s \succ_{F(\Pi)} v.
\end{equation}
Thus, by {\sf IIA}, if  $s \in S$ is
unanimously preferred to $v \notin S$,  $s$ must be
strictly preferred to $v$ in the aggregate preference.
	
	Now let $\pi'$ be the same as $\pi$ only with the least preferred member of $S$, $s'$, and the most preferred non-member, $v'$, switched in rank. By the partial {\sf Unanimity} property \eqref{partial-unam1}, in the aggregate $F(\Pi')$, all members of $S-\{s'\}$ are preferred to all $v\notin S$, and all members of $S$ are preferred to all $v\in V-S-\{s'\}$.  On the other hand, by
{\sf GS}, $S\notin \calC_F(\Pi')$, which is only possible is if $v' \succeq_{F(\Pi')} s'$.  Applying the partial {\sf Unanimity} property once more yields the following two statements:
$$
\forall s \in S-\{s'\}, s \succ_{F(\Pi')} s' \quad \text{ and }\quad
	\forall v \notin S \cup \{v'\}, v' \succ_{F({\Pi'})} v,
$$
{ and by {\sf IIA}, this in turn implies}
\begin{equation}
{\label{partial-unam2}
\forall s \in S-\{s'\}, s \succ_{F(\Pi)} s' \quad \text{ and }\quad
	\forall v \notin S \cup \{v'\}, v' \succ_{F(\Pi)} v.}
\end{equation}
By  {\sf IIA}, this means that for any two members or two non-members if one
is unanimously preferred to the other, then it must be strictly preferred in aggregate preference.
Indeed, consider, e.g., $s,s'\in S$ and a profile $\tilde\Pi_S$ such that $s\succ_{\tilde\pi_i} s'$
for all $i\in S$.  Choose $\Pi$ in such a way that every member has the same profile, $s'$ has rank $|S|$ and $s\succ_{\pi_i} s'$ for all $i\in S$.  By  {\sf IIA}, $s\succ_{F(\tilde \Pi)} s'\Longleftrightarrow s\succ_{F( \Pi)} s'$, so by \eqref{partial-unam2}, $s$ is preferred to $s'$ in aggregate.

	Finally, consider $\pi''$ where $v'$ is switched with the
        second lowest ranked member, $s''$.  By the above additional
        partial {\sf Unanimity} property, ${s'}$ must be strictly
        preferred to ${s''}$ in the aggregate preference {$F(\Pi'')$},
        and therefore $v'$ must be strictly rather than weakly
        preferred to $s''$ in the aggregate preference.  Thus, again
        by {\sf IIA}, if a non-member, $v \notin S$, is unanimously
        preferred to a member $s \in S$, $v$ must be strictly
        preferred to $s$ in the aggregate preference.  Taken together,
        these three partial {\sf Unanimity} properties, constitute
        {\sf Unanimity}.
	
	Since  the election $(V,F,S)$ satisfies both
{\sf IIA} and {\sf Unanimity}, by Arrow's Impossibility Theorem \cite{ArrowBook} it must be a dictatorship, contradicting  Proposition~\ref{prop:ND}.
\end{proof}

\subsection{Additional Properties of Axioms}

Here we state some additional properties of interest that community
rules (not necessarily fixed point rules) have when they satisfy one
or more of our main axioms.

\begin{proposition}
\label{prop:Clique-Null}
	Let $\calC$ be a community rule that satisfies the {\sf World Community} and {\sf Embedding} Axioms. Then $\calC$ must also satisfy the {\sf Cliques} 
Property.
\end{proposition}

\begin{proof}
	Let $A = (V,\Pi)$ be a preference network and $S$ be a clique (every member of $S$ prefers $S$ to $V-S$).  By {\sf World Community}, we have that $S \in \calC((S, \Pi|_S))$ and by {\sf Embedding} we have $\calC(A) \cap 2^S = \calC((S, \Pi|_S))$.  Therefore $S$ is a community.
\end{proof}

\begin{proposition}
Any community rule $\calC$ that satisfies  {\sf Monotonicity}
must satisfy  {\sf Independence of Outside Opinions}.
\end{proposition}
\begin{proof}
	Let $A=(V,\Pi)$ be a preference network.  Axiom {\sf Mon} features an alternative preference profile $\Pi'$ stating that if $\Pi'$ satisfies certain properties and $S$ is a community for $(V,\Pi')$, then $S$ must be a community $(V,\Pi)$.  Because the axiom places no restrictions on the preferences of voters from $V-S$, the rule $\calC$ must satisfy {\sf IOO}.
\end{proof}

\begin{property}[{\sf Outsider Departure (OD)}]
	A community rule $\calC$ satisfies the {\sf Outsider Departure} Property if for a given preference network $A=(V,\Pi)$, community $S \in \calC(A)$, and outsider $v \notin S$, we have that $S \in \calC(V-\{v\}, \Pi|_{V-\{s\}})$.
\end{property}

\begin{proposition}
\label{prop:MonEb->OD}
	A community rule, $\calC$, that satisfies the {\sf Monotonicity} and {\sf Embedding} Axioms must also satisfy the {\sf Outsider Departure} Property.
\end{proposition}
\begin{proof}
	Let $A =(V,\Pi)$ be a preference network, $S \in \calC(A)$ a community, and $v \notin S$ an outsider.  Consider the preference profile $\Pi'$ that ranks $v$ at the end of everyones preference.  By {\sf Mon}, $S \in \calC(V, \Pi')$.  Furthermore, since $\Pi'$ satisfies the setup for {\sf Embedding}, we also have $S \in \calC(V-\{v\}, \Pi'|_{V-\{v\}})$.  However, $\Pi'|_{V-\{v\}} = \Pi|_{V-\{v\}}$ since $\Pi$ and $\Pi'$ only differ in the placement of $v$.  Therefore we have $S \in \calC(V-\{v\}, \Pi|_{V-\{s\}})$.
\end{proof}

\begin{proposition}
	If a community rule satisfies the {\sf Group Stability} and {\sf Self-Approval} Axioms it must satisfy the {\sf Pareto Efficiency} Property.
\end{proposition}
\begin{proof}
	Let $S$ be a community.

\vspace{0.03in}
 Case 1: $|S| = 1$.
			By {\sf Self-Approval}, the one member $s$ must rank herself above all outsiders and therefore satisfies {\sf PE}.

\vspace{0.03in}
 Case 2: $|S| > 1$.		
			Choose $G \subset S$ such that $G$ is a singleton $\{s'\}$.  By {\sf Group Stability}, for all outsider singletons $\{g'\} \subseteq V-S$ and bijections $(f_i:\{s'\} \rightarrow \{g'\}, i \in S-G)$ there exists an $s \in S-G$ such that $s' \succ_{\pi_s} f_s(s')$.  Since it is clear that $f_s(s') = g'$, $s$ provides the necessary witness for $s'$ and $S$ satisfies {\sf PE}.
\end{proof}

\section{Aggregation Based Communities Rules}\label{Sec:Rules}
We now examine several examples of aggregation based community rules
 through the lens of our axiomatic framework.
In Section \ref{sec:SCRules},
we focus on a what we call
weighted fixed-point rules,
 starting with the \B3CT community function from \cite{B3CT}.
We
show that it violates both Axioms
  {\sf Monotonicity} and {\sf Group Stability}.
The violation of the monotonicity axiom was initially
  somewhat of a surprise and rather counterintuitive to us.
This violation is illustrative of the
  subtlety of community rules;  indeed, it helped us to
  identify a weaker monotonicity property that
  the \B3CT function satisfies.
We then show that the fixed-point community rule based on
  any Borda-count-like voting function
  is inconsistent with either the {\sf Group Stability} axiom
  or the {\sf Clique} property.
This impossibility result and Theorem \ref{thm:impossibility}
   illustrate some basic limitations of  fixed-point community rules.
Next, we study the properties of the harmonious community function
in Section~\ref{sec:Harmonious}.  We will show that it can be
obtained by preference aggregation, and that it obeys all of our axioms
except for  Axiom {\sf GS}.  It does, however, satisfy a weaker version
of this axiom, see Theorem~\ref{theo:H}.  In our final subsection,
Section~\ref{sec:comp-of-Borda-etc},
we compare
the three rules
  Borda voting, \B3CT voting, and the harmonious rule.

\subsection{Weighted Fixed Point Rules} \label{sec:SCRules}

This section focuses on a class of community rules that lie in between
general fixed point rules and the \B3CT community rule, which we call
{\em weighted fixed point rules}.  First, we will look at some of the
properties of the \B3CT rule as a particular case of a weighted fixed
point rule.

\begin{theorem}\label{th:B3CT}
The \B3CT community rule, $\calC_{B^3CT}$, does
  not satisfy {\sf Monotonicity} or
{\sf Group Stability}.
It satisfies all other axioms, as well as  Properties {\sf Pareto Efficiency} and {\sf Clique}.
\end{theorem}
\begin{proof}
Directly from the definition of the \B3CT voting function $\phi_S^\Pi$,
  $\calC_{B^3CT}$ satisfies Axioms {\sf A}, {\sf WC},  {\sf Emb},
and Properties  {\sf PE} and {\sf Cq}.
Suppose $\calC_{B^3CT}$ does not  satisfy
  {\sf SA}.
Then, there exists a preference network $A = (V,\Pi)$,
  $S \in \calC_{B^3CT}(A)$,  $T \subseteq V-S$,
   and a tuple of bijections  $(f_s: S\rightarrow T)$
  such that for all $s,u\in S$, $u\prec_{\pi_s}
  f_s(u)$.
It follows that $\forall s\in S$, the numbers of votes cast by
  $s$  for $S$ according to $\phi_S^{\Pi}$ is less than the
  numbers of votes that $s$ casts for $T$.
Summing up the votes from $S$,
  the average votes that members of $T$ receive is larger than the
  average votes that members of $S$ receive, contradicting
  the assumption that everyone in $S$ receives more votes than
  everyone in $T$.
Thus, $\calC_{B^3CT}$ satisfies {\sf SA}.

To show $\calC_{B^3CT}$ satisfies Axiom {\sf CRM},
consider $S$,  $\Pi$ and $\Pi'$ as in Axiom {\sf CRM}.
By the coherence assumption for members,
there exists
    $\sigma\in L(S)$
    such that for $s_1,s_2\in S$, for all $s\in S$,
    $s_1 \succ_{\pi'_s} s_2$ if  and only if $ s_1 \succ_{\sigma} s_2$.

Let $s^*$ denote the least preferred elements of $S$ according to  $\sigma$.
By the assumption that $\pi_s(v)=\pi'_s(v)$ for all
    $s\in S,v\in V-S$, we have that $\pi_s(V-S)=\pi'_S(V-S)$, and hence
    also that $\pi_s(S)=\pi'_s(S)$.  But this implies that for all
    $u\in S$

\[
    \phi_S^\Pi(u)=\sum_{s\in S} 1_{\pi_s(u)\leq |S|}
    \geq \sum_{s\in S} 1_{\pi_s(S)\subseteq [1:S]}
    = \sum_{s\in S} 1_{\pi'_s(S)\subseteq [1:S]}
    =\sum_{s\in S} 1_{\pi'_s(s^*)\leq |S|}
    =\phi_S^{\Pi'}(s^*).
\]
If $S\in \calC_{B^3CT}(V,\Pi'))$,
 then $s^*$ receives more votes from $\Pi'_S$
than every $v\in V-S$, and
the number of votes $v$ receives from $\Pi_S$ is the same
 as the number of votes it receives from $\Pi'_S$.
On the other hand, for all $u\in S$, the number of votes $u$ receives from
$\Pi_S$ is at least the number of votes $s^*$ receives from
$\Pi'_S$, implying that $S\in \calC_{B^3CT}(V,\Pi)$.
We can similarly show that $\calC_{B^3CT}$ satisfies Axiom {\sf CRNM}.

Let $V = [1:6]$, $S=[1:3]$, let $\Pi = (\pi_1,...,\pi_6)$ be
the preference profile
\begin{eqnarray*}
\pi_1 = [1 4  2 3 5 6],\quad \pi_2 = [{2 5 }3 4 1 6],\quad \pi_3 =
[6 3 1 4 2 5]\\
\pi_4 = [4 5 6 1 2 3], \quad \pi_5 = [1 5 6 4 2 3], \quad
\pi_6 = [1 6 5 4 2 3]
\end{eqnarray*}
and let $\Pi'$ be the preference profile
\begin{eqnarray*}
\pi'_1 = [1 {4 2} 3 5 6], \quad\pi'_2 = [2 3 4 5 1 6], \pi'_3 =
[3 1 4 6 2 5]\\
\pi'_4=\pi_4,\quad\pi'_5=\pi_5,\quad \pi'_6 = \pi_6.\qquad
\end{eqnarray*}
Then $S=[1:3]\in\calC_{B^3CT}(V,\Pi)$,
 as each members of $S$ receives two votes
  while everyone in $[4:6]$ receives only one vote.
However, in violation of Axiom {\sf Mon},
  $S $ is no longer a \B3CT community
{w.r.t $\Pi'$, since $4$ now receives three votes, one more than $1$, $2$ and $3$.}

Note also $T  = (1, 5, 6)\in \calC_{B^3CT}(V,\Pi)$.
Let $G = \{5,6\}\subset T$ and $G' = (2, 4)\subset V-T$.
As member $1$ prefers $2$ to $5$ and $4$ to $6$, $T$ does not satisfy
{\sf Group Stability}.
\end{proof}

Note that  the same analysis shows that $\calC_{B^3CT}$
  does not satisfy the {\sf Outsider Departure} Property.
In the example above, if member $5$ leaves the system, then member $4$
  will receive 2 votes from $S = \{1,2,3\}$, and hence $S$ is
   no longer a $\calC_{B^3CT}$-community.

Even though $\calC_{B^3CT}$ does not satisfy {\sf Mon},
  it does enjoy the following monotonicity property.
\begin{property}[{\sf Outsider Respecting Monotonicity}]
If $S$ is  a community of a preference network   $A=(V,\Pi)$,
 then   $S$ remains a community of $(V,\Pi')$ for
   any $\Pi'$ such that (1)
 $u \succ_{\pi_s} t$ $\Rightarrow$ $u \succ_{\pi'_s} t$, $\forall
 u,s\in S, t \in V$, and (2) $v \succ_{\pi_s} v' \Rightarrow
 v \succ_{\pi'_s} v'$,
 $\forall v,v'\in V-S, s\in S$.
\end{property}

We now analyze the fixed point rule defined by
  the family of aggregation functions, such as
  Borda count and \B3CT voting, that derive a cardinal social preference
  from ordinal individual preferences.

Let $W$ be a  sequence of weight vectors $w^i\in {\mathbb R}^{n}$, $W=(w^1,w^2,\dots)$, where $n$ is
the number of elements in $V$.
 For a non-empty finite $S\subset{\mathbb N}$ and  $\Pi_S\in \LV^S$
 define the aggregate preference $F_W(\Pi_S)$ on $V$ by
$$
i\succ_{F(\Pi_s)} j\qquad\Longleftrightarrow
\qquad
\sum_{s\in S} w^{{ |S|}}_{\pi_s(i)}
>
\sum_{s\in S} w^{{ |S|}}_{\pi_s(j)}.
$$	
In other words, $i\succ j$ in the aggregate iff the total weight of the votes $i$ receives from $S$ is larger than
the total weight of the votes $j$ receives from $S$, where a vote in position $p$ gets weight $w_p^{|S|}$.

In \B3CT, $w^k$ is the vector of $k$ ones followed
by $(n-k)$ zeros\footnote{The rule $\calC_{\text{\B3CT}}$ does not specify what the weight $w^k$ should be
for $k>n$ since preferences with more voters than alternatives do not occur
when determining communities -- so we are free to define it arbitrarily, say
$w^k_i=1$ for all $i$ if $k>n$.}, while
Borda count uses $w^k = (n,n-1,...,1)$ for all $k$.

\begin{definition}[Weighted Fixed Point Rule]\label{def:FPR}
For a sequence of vectors $W = (w^1,w^2,\dots)$ in $\mathbb{R}^n$,
 $\mathcal{C}_W$ is the fixed point rule with respect to $F_W$.
\end{definition}

\begin{proposition}
Weighted fixed-point rules satisfy Axiom {\sf Anonymity}.
They satisfy
{\sf Outsider Respecting Monotonicity} if $w_i^k\geq w_j^k$ for all $k\in [1:n-1]$
and $i\leq j$, and they satisfy
the {\sf Clique} Property if and only if
for all $k\in [1:n-1]$ the weight vector $w^k$ is such that $w^k_i > w^k_j$ for $i \le k$ and $j > k$.
\end{proposition}
\begin{proof} {The proof of the first two statements and the ``if'' part of the third follow directly from the definitions.
To see the ``only if'' part of the third statement, consider $k,i,j$ such that $w^k_i \leq w^k_j$, and let $S,\Pi\in L(V)$ be such that $|S|=k$,
  $\pi_s\leq k$ for all $s,u\in S$,
and $\pi_s(v)=\pi_t(v)$ for all $s,t\in S, v\in V$.
Then $S$ satisfies the condition of the Property {\sf Cq}, but it is not a community.  To see this, choose $v\in S$ and $v'\notin S$
such that $\pi_s(v)=i$ and $\pi_s(v')=j$.  Then $\sum_{s\in S}w^k_{\pi_s(v)}\leq \sum_{s\in S}w^k_{\pi_s(v')}$, showing
that $S$ is not a community.}
\end{proof}

{Together with Proposition~\ref{prop:Clique-Null}, the next theorem implies that there is no weighted fixed
point rule that satisfies the {\sf Group Stability}, {\sf World Community} and {\sf Embedding} Axioms.}

\begin{theorem}{\sc (Impossibility of Weighted Aggregation Schema)}
Weighted Fixed Point Rules are inconsistent with either the {\sf Group
Stability}  Axiom or the {\sf Clique} Property.
\end{theorem}	
\begin{proof}
Let $A=(V,\Pi)$ be a preference network, $S \subset V$,
and $\calC_W$ a weighted fixed point rule satisfying the
the {\sf Clique} Property.
Throughout the the proof, we will take
$$
V = \{a,b,c,d,e\}\quad\text{and}\quad S = \{a,b,c\},
$$
{ and consider preference
profiles such that $S$ violates {\sf Group Stability}.
In order for $\calC_W$ to obey the Axiom {\sf GS}, we would
need  the weight vector $w^3\in\mathbb R^5$ to be such
that  $S\notin\calC(V,\Pi)$ for all $\Pi$ considered in this proof.
Our goal is to show that this will lead to a contradiction.
We start under the assumption that the weights are decreasing, i.e., in addition
to the already established fact that $w_i^3>w_j^3$ when $i=1,2,3$ and $j=4,5$ (since $\calC_W$  satisfies the
the {\sf Clique} Property)},
we will first assume that $w_1^3\geq w_2^3\geq w_3^3$ and $w_4^3\geq w_5^3$.

Consider the following scenario:
\begin{eqnarray*}
\pi_a = [adebc] ,\ \pi_b = \pi_c =[abcde].
\end{eqnarray*}
{Since $a$ prefers $d$ and $e$ over
 $b$ and $c$, $S$ is not group stable} and hence cannot be a community.
 By our assumption that $w_1^3 \ge w_2^3 \ge w_2^3 > w_4^3 \ge w_4^3$, we have that $a \succ_{F_W(\Pi_S)} b \succeq_{F_W(\Pi_s)} c \succ_{F_W(\Pi_S)} e$ and  $b \succ_{F_W(\Pi_S)} d$.  Therefore the only way $S$ cannot be a community is that
  $d \succeq_{F_W(\Pi_S)} c$, i.e.,
$$
w_2^3+2w_4^3\geq 2w_3^3+w_5.
$$
Notice that this implies that we cannot have both $w_2^3 = w_3^3$ and $w_4^3 = w_5^3$.
	
	Now consider a modified preference profile:
$$\pi'_a =\pi'_b = [abdce],\ \pi'_c ={[caebd]}.$$

In this profile $a$ and $b$ prefer $d$ over $c$, so again $S$ violates {\sf GS} and hence cannot be a community.
On the other hand, we now have $a \succ_{F_W(\Pi'_S)} b$, $b \succeq_{F_W(\Pi'_S)} d \succ_{F_W(\Pi'_S)} e$.
 Thus we must have either $b \sim_{F_W(\Pi')} d$ or
 $d \succeq_{F_W(\Pi'_S)} c$.
The former, however, implies $w_2^3 = w_3^3$ and $w_4^3 = w_5^3$ and is hence a contradiction.  Therefore the latter must be true
which implies
$$
 2w_3^3 + w_5^3\geq w_1^3 + 2w_4^3 .
$$

	This brings us to the final preference profile:
$$\pi_a'' = [abdce],\ \pi_b'' = [dcabe],\ \pi_c'' = [cbaed].$$
{Again $a$ and $b$ prefer $d$ to $c$, so the profile violates {\sf GS}, and hence again can't be a community.
Now $a \succ_{F_W(\Pi'')} c\succeq_{F_W(\Pi'')}b $ and $d\succ_{F_W(\Pi'')} e$, showing that
for $S$ not to be a community, we must have $d\succeq_{F_W(\Pi'')} b$, which gives
$$
w_1^3+w_3^3+w_5^3\geq 2 w_2^3+w^3_4.
$$

Defining $d_i=w_i^3-w_{i-1}^3$, we can write the bounds obtained so far as
$$
d_4\leq d_3+d_5
$$
$$
d_2+d_3+d_5\leq  d_4
$$
$$
d_3+d_4+d_5\leq  d_2.
$$
Chaining up these three bounds, we get
$$
d_3+d_5\geq d_4\geq d_2+d_3+d_5\geq d_3+d_4+d_5+d_3+d_5=2(d_3+d_5)+d_4,
$$
contradicting our assumption $d_i\geq 0$ and the fact that {\sf Cq}  implies $d_4>0$.

To relax the constraint that the weights are ordered, we observe that all three profiles considered in the proof
are such that, under arbitrary permutations of the first three and the last two positions,  $S$ still violates {\sf GS}.  In other words, for any permutation $\sigma$ of $[1:5]$ that leaves $[1:3]$ and $[4:5]$ invariant, $S$ violates
{\sf GS} under the profiles $\{\sigma \circ\pi_s\}_{s\in S}$, $\{\sigma' \circ\pi_s\}_{s\in S}$, and
$\{\sigma'' \circ\pi_s\}_{s\in S}$.
Choosing the permutation in such a way that the weights $\tilde w^3_i=w^3_{\sigma(i)}$ are ordered, we obtain
the above three inequalities for the weights $\tilde w_i^3$, leading again to a contradiction.
}
\end{proof}

\subsection{Properties of Harmonious Communities}
\label{sec:Harmonious}

In this subsection, we analyze the harmonious community function given by
  Definition \ref{def:harmonious}.  We first prove that it can be expressed
  in terms of a suitable preference aggregation function.

 \begin{proposition}\label{prop:H}
 \label{prop:harmonious-aggregation}
 There exists a preference aggregation function
   $F_{\calH}:\LV^* \rightarrow \bLV$ such that the harmonious community function $\calH$
 is defined by a  $F_{\calH}$.
 \end{proposition}

 \begin{proof}
 Given $V$, a finite set $S$, and a preference profile $\Pi_S\in L(V)^S$, we
consider the following  directed graph $G_{\Pi_S} = (V,E_{\Pi_S})$
   where $(i,j)\in E_{\Pi_S}$ if at
   least half of $S$ prefers $i$ to $j$.
Note that if $|S|$ is an odd number, then
 $G_{\Pi_S}$ is a {\em tournament graph.}
If $|S|$ is an even number, then $E_{\Pi_S}$  contains
  both $(i,j)$ and $(j,i)$ if exactly half of $\Pi_S$ prefer $i$ to $j$.
$G_{\Pi_S}$ is {\em total} since for all $i,j\in V$, either
$(i,j)\in E_{\Pi_S}$ or $(j,i)\in E_{\Pi_S}$.
Because $G_{\Pi_S}$  is total, the graph $\hat G_{\Pi_s}$
{ obtained from $G_{\Pi_S}$
by contracting each
strongly connected component into a single vertex} is an {\em acyclic}, tournament graph.
As a consequence, the graph $\hat G_{\Pi_s}$  has exactly one Hamiltonian path that
  totally orders its vertices.
Let  $(V_1,...,V_t)$ be the strongly connected components of $G_{\Pi_S}$,  sorted by
{ the order determined by the Hamiltonian path.  The partition
 $(V_1,...,V_t)$ of $V$ then defines an ordered partition $F_{\calH}(\Pi_S)$, with $V_i\succ_{F_{\calH}(\Pi_S)} V_j$
iff $i\leq j$.

Next, we consider a subset $T\subset V$.  It is then easy to check that
if $T$ is of the form $T = \cup_{j\leq i} V_j$ for some $i\in [1:t]$, then
for all $u\in T, v\in V-T$, a majority of $S$ prefers $u$ to $v$, and vice versa.
Specializing to $S=T$, we see that}
$\calH$ is defined by the preference aggregation function
$F_{\calH}$.
\end{proof}

Next we  show that $\calH$
satisfies all  axioms except for {\sf Group Stability}.

\begin{theorem}\label{theo:H}
The harmonious community function
  satisfies Axioms
{\sf A}, {\sf SA}, {\sf Mon}, {\sf Emb}, {\sf WC}, {\sf CRM}, and {\sf CRNM}, but it
does not
satisfy {\sf GS}.
\end{theorem}

\begin{proof}
Directly from the definitions,
one easily checks that $\calH$ satisfies Axioms {\sf A}, {\sf Mon}, {\sf Emb} and {\sf WC}.

By a similar argument to the proof of Theorem \ref{th:B3CT}, we can
  prove that $\calH$ satisfies {\sf SA}: if $S\in \calH(A)$ does not satisfy {\sf SA}, then
  {there exists a $T\subset V-S$ of the same size as $S$ such that  each $s\in S$
  lexicographically prefers $T$ over $S$.  With the help of Proposition~\ref{prop:lex},
  this implies that, for each $s\in S$, there are at least $(1+2+\dots+|S|)$ pairs $(u,v)\in S\times T$
  such that $s$ prefers $v$ over $u$. } Thus the number of triples
  $(s,u,v)$ such that $s\in S$ prefers $v\in T$ over $u\in S$ is at least
  $|S|^2(|S|+1)/2$.
  However, $S\in \calH(A)$ implies that this number has to be
 strictly smaller than $|S|^3/2$.

To see that $\calH$ is consistent with Axiom {\sf CRNM},
consider a preference profile $\Pi,\Pi'$ as specified in Axiom {\sf CRNM}.
By the coherence assumption on non-members, there exists a linear order $\sigma$ on $V-S$,
  such that $\forall i,j\in V- S$ and $\forall s\in S$,
  $i \succ_{\pi'_s} j \Leftrightarrow i \succ_{\sigma} j$.
Let $v^* $ be the most preferred element of $\sigma$.
{ By the assumption that
$\pi_s(u)=\pi'_s(u)$ for all $s,u\in S$, we have $\pi_s(S)=\pi'_s(S)$
and hence also $\pi_s(V-S)=\pi'_s(V-S)$.  But this implies
that for all $v\in V-S$,
\[
\pi_s(v)\geq \min\{i\in \pi_s(V-S)\}=\min\{i\in \pi'_s(V-S)\}=\pi'_s(v^*).
\]
We therefore have shown that for all
$s,u\in S$ such that
$u \succ_{\pi'_s} v^*$,  we have that
$ u \succ_{\pi_s} v$ for all $ v\in V-S$.
}
Assume now that $S\in \calH((V,\Pi'))$.  Then for all $u\in S$, the majority of
 $(\Pi',S)$ prefer $u$ to $v^*$, which, as we just have shown, implies
 that for all $v\in V-S$, the majority of $(\Pi,S)$  prefer $u$ to $v$,
 which in turn implies that $S\in \calH((V,\Pi))$.
We can similarly show that $\calH$ satisfies Axiom {\sf CRM}.

The set $T$ in the proof of Theorem \ref{th:B3CT} is also an example
  that $\calH$   violates Axiom {\sf GS}.
\end{proof}

While $\calH$ does not satisfy the {\sf GS} Axiom, it satisfies the following weaker property.

\begin{property}{\sc {\sf Weak Group Stability }}
{ For all preference profiles $\Pi$ on $V$ and all
$S\in\calC(V,\Pi)$, $S$ is {\sf  weakly group stable}.
Here a set $S\subset V$ is called  weakly group stable} if
for all  $ G \subset S$, $G' \subset V - S$ s.t.
$0  < |G| = |G'| \leq |S|/2$,
{ and all  bijections $(f: \
G \rightarrow G', i\in S-G)$} there exists $s \in S - G$, $u \in G$
such that $u \succ_{\pi_s}  f(u)$.
\end{property}

{Note that the property is weaker than the {\sf GS} Axiom in two ways: we restrict
ourselves to groups $G$ of size at most $|S|/2$, and we only allow for a global
bijection $f$, rather than individual bijections $f_s$.

\begin{proposition}
$\calH$  is  {\sf weakly group stable}, while the Borda count and the  \B3CT rule are not.
\end{proposition}

\begin{proof}
Consider a set $S\in \calH(V,\Pi)$, subsets
 $G\subset S$ and $G' \subset V - S$ such that $0 <
  |G| =|G'| \leq |S|/2$, and a bijection $f: G\rightarrow G'$.
For each $u\in G$ the  majority
of $S$ prefer $u$ to $f(u)$   (who is not a member of $S$), and since
$|G|\leq |S|/2$, this implies that there must be at least one $s\in S-G$
such that $s$ prefers $u$ to $f(u)$, as required.

To give a counterexample for both Borda counting and the \B3CT rule,
consider $V=[1:6]$, $G=[3:4]$ and $G'=[5:6]$, with preference profiles
\[\pi_1=[125463],\,\pi_2=[126354],\, \pi_3=[341256], \pi_4=[341256].
\]
Then $1$ and $2$ prefer $5$ over $4$, and $6$ over $3$, but
$S$ is a community both with respect to \B3CT (where $1$ and $2$ get four votes,
$3$ and $4$ get three votes, and $5$ and $6$ get only one vote),
and with respect to Borda count (with counts $20, 16, 18, 16,10,8$ for $1,\dots,6$,
respectively).
\end{proof}

\begin{proposition}
$\calH$ satisfies {\sf IOO} 
as
well as {\sf Cq} and the {\sf PE}, but $F_{\calH}$ does not satisfy {\sf U}.
\end{proposition}

\begin{proof}
By Proposition~\ref{prop:fix-point-props}, $\calH$ satisfies {\sf IOO}.
To see that it does not satisfy  {\sf U}, let $V=\{a,b,c\}$,
 let $S=\{a,b\}$
  and $\pi_a=(acb)$, $\pi_b=(bac)$.
  Then $a\succ_{\pi_s} c$ for all $s\in S$, and both
  $a\succ_{\pi_s} b$ and $b\succ_{\pi_s} c$ in  half of $S$.
Therefore $(ac), (cb),(bc),(ab),(ba)\in E_{\Pi_S}$.
Thus, $a,b,c$ belongs to the same connected component in
  $G_{\Pi_S}$, showing that $S$ is not a harmonious community.  To see
  that $\calH$ satisfies both {\sf Cq} and {\sf PE} in spite of the fact
  that it does not satisfy the
assumptions of Proposition~\ref{prop:Pareto},  we use
Proposition~\ref{prop:Clique-Null} to infer {\sf Cq}, and the observation that
 $S\in \calH(A)$ implies that for any a pair of elements $(u\in S,
  v\not\in S)$, the majority of $S$ prefer $u$ over $v$, proving {\sf PE}.
\end{proof}
}

\subsection{Comparison of Borda voting, \B3CT voting, and the harmonious rule}
\label{sec:comp-of-Borda-etc}
In this subsection, we compare
  the fixed-point community rules that we have discussed so far:
  Borda voting, \B3CT voting, and the harmonious rule.
While all three have their own appealing simplicity and intuition and
  all satisfy Axioms {\sf A}, {\sf SA}, {\sf Emb}, {\sf WC},
  {\sf CRM}, and {\sf CRNM},
  there are significant differences with respect
  to Axioms {\sf Mon} and {\sf GS}, and
  the {\sf Outsider Departure} property.

\vspace{0.05in}

\begin{itemize}
\item   {\sf Outsider Departure}:
A harmonious community $S$ remains a harmonious community when
  any outsider $v\not\in S$ leaves the system since
  the departure does not alter any pairwise preferences.
However, for a \B3CT community $S$, the departure of an outsider can increase the votes for
  other outsiders enough to destabilize the \B3CT community.
In a similar way, one can see that the Borda count rule is also unstable to
departure of an outsider.

\item {\sf Monotonicity}: The harmonious rule satisfies
  Axiom {\sf Mon}.
The other two only satisfy the weaker {\sf Outsider Respecting
  Monotonicity} property%
  \footnote{Again, we can use the profiles from the proof of Theorem~\ref{th:B3CT}
  to show that the Borda count rule does not satisfy {\sf Mon}.}.

\item {\sf Group Stability}: The subset $T$ in
  the proof of Theorem \ref{th:B3CT}
  is a community according to all these three community rules.
But $T$ violates {\sf GS} because 1 prefers
  outsiders over 5 and 6, even though 5 and 6 prefer 1
  over everyone else: Element 1 is an
  {\em ``arrogant'' member} of its community.
All aggregation functions
   satisfying {\sf Unanimity} seem to be prone to existence of
  ``arrogant''  members.
The harmonious
  rule satisfies the stability of majority subgroup under a global
  bijection $f$,
  although the stability of the minority subgroup
(or the majority subgroup with
  individual bijections $f_s$)
  may not be
  guaranteed.
The fixed-point rule of  Borda count  and \B3CT voting
  essentially have no guarantee of group stability.

\item {\sf Small World}:
In general, we say a community function $\calC$ satisfies
  the {\sf Small World} property if
\begin{eqnarray*}\label{eqn:smallworlds}
S\in \calC((V,\Pi)) \Leftrightarrow
 \forall U\subseteq V-S, |U|< |S|, S\in \calC(S\cup U, \Pi|_{S\cup U}).
\end{eqnarray*}
This Helly-type property \cite{Helly} localizes the identification of a
community.  Note that the {\sf Small World} property includes some form of
  {\sf Outsider Departure} together
  with the property that every community is ``locally'' verifiable.
One can easily show that the fixed-point rules of the Borda count
  or \B3CT voting do not have the {\sf Small World} property,
  while the harmonious rule  enjoys the following
  stronger variant of the small world property
\begin{eqnarray*}\label{eqn:smallworldsHarm}
S\in \calH((V,\Pi)) \Leftrightarrow
 \forall v \in V-S,  S\in \calH(S\cup \{v\}, \Pi|_{S\cup \{v\}}),
\end{eqnarray*}
and hence the property
  given in (\ref{eqn:smallworlds}).

\end{itemize}

\section{Taxonomy of Community Rules}\label{sec:taxonomy}

In this section,
we characterize the taxonomy
of the axiom-conforming community rules.

First, in Section~\ref{sec:tax-theorem}, we define two rules, the
Clique Rule and the Comprehensive Rule, which satisfy all axioms, and
which are most selective and most comprehensive, respectively, in the
sense that any rule which satisfies all axioms leads to a set of
communities which contains all communities defined by the Clique Rule
and is contained in the Comprehensive Rule (the statement that this is
the case, Theorem~\ref{th:taxonomy}, will be our main theorem in this
subsection).

In the next subsection, Section~\ref{sec:lattice}, we
then expand on this ``Taxonomy Theorem'', and show that under the
following natural intersection and union of community rules,
 the family of all community rules that
satisfies all eight axioms forms a bounded lattice.
We will use the following two set-theoretic operators
  of community functions to define these lattice structures.

\begin{definition}[Operations over Community Rules]
For two community functions $\calC_1$ and $\calC_2$, we define the intersection and union,
$\calC_1\cap \calC_2$ and $\calC_1\cup \calC_2$, as the community functions which, for all
preference networks $A$, respectively satisfy
$$
\begin{aligned}
(\calC_1\cap \calC_2)(A)& := \calC_1(A)\cap \calC_2(A)\\
(\calC_1\cup \calC_2)(A)&: = \calC_1(A)\cup \calC_2(A).
\end{aligned}$$
\end{definition}

\subsection{From the Most Selective to Most Comprehensive Rule}
\label{sec:tax-theorem}

We start with perhaps the simplest rule for communities that satisfies
  the {\sf Clique} Property.
\begin{Rule}[Clique Rule  ($\calC_{clique}$)]
A  non-empty subset
 $S\subseteq V$ is a community of  $A = (V,\Pi)$, if and only if
 $\forall u,s \in S$, $v \notin S$, $u \succ_{\pi_s} v$.
We use $\calC_{clique}$
  to denote the community function defined by this rule.
\end{Rule}

\begin{proposition}
	$\calC_{clique}$ satisfies all Axioms.
\end{proposition}

\begin{proof}
The (easy) proof is left as an exercise for the reader.
\end{proof}

However, the clique rule appears to be too restrictive, since it has the
following structural feature, which essentially rules out any non-trivial overlap of communities, while
 ``Real-world'' communities typically have non-trivial overlaps among themselves.

\begin{proposition}
For any preference network $A$, if $S_1, S_2\in \calC_{clique}(A)$, then either $S_1\cap S_2
= \emptyset$ or  $S_1 \subset S_2$, or $S_2 \subset S_1$.
\end{proposition}

\begin{proof}
	Assume otherwise.  By assumption, we can choose an element $s \in S_1 \cap S_2$.  Without loss of generality assume $|S_1| \le |S_2|$.  Again by assumption, there exists an element $s' \in S_1$ and $s' \notin S_2$.  By the definition of the $\calC_{clique}$ $s$ must have $s'$ in its top $|S_1|$ choices.  However, this means that $s'$ is also in the top $|S_2|$ choices for $s$, which violates the fact that $S_2$ is in $\calC_{clique}(A)$.
\end{proof}

Next we address the  question of whether there are
rules consistent with all axioms that
admit overlapping communities.
To address this question,
we consider rules defined by community axioms.

\begin{Rule} [Axiom Based Community Rules]
For
  $X\in \{ \mbox{\sf  GS, SA}\}$
let $\calC_{X}$ be the community rule defined by
$A = (V,\Pi)\mapsto \calC_X(A)$, where
$\calC_{X}(A)$ is the set of non-empty subsets $S\subset V$ such that
$S$ obeys axiom {\sf X}.
\end{Rule}

For example, $\calC_{GS}$ denote the  community rule that
  $S\in \calC_{GS}(A)$ if and only $S$ enjoys the {\sf Group
  Stability} Axiom.

The first part of our Taxonomy Theorem is a direct consequence
  of the following basic lemma.

\begin{lemma}{\sc (Intersection Lemma: {\sf GS} and {\sf  SA})}\label{lem:intersection}
For  $X\in \{ \mbox{\sf A, Mon, CRM, CRNM, WC, Emb}\}$,
  if  $\calC$ satisfies Axiom $X$,
  then $\widetilde{\calC} = \calC \cap \calC_{GS}\cap \calC_{SA}$ satisfies
Axioms $X$, {\sf GA} and {\sf SA}.
\end{lemma}

\begin{proof}
$\calC_{GS}$ and $\calC_{SA}$
  are both consistent with  {\sf A}, {\sf WC}, and {\sf Emb}, thus
  if $\calC$ satisfies Axiom $X\in \{ \mbox{\sf A,  WC, Emb}\}$,
   then
  $\widetilde{\calC}$ remains consistent with Axiom X.

To see $\widetilde{\calC}$ satisfies Axiom {\sf Mon} if $\calC$ satisfies {\sf Mon},
  choose $\Pi,\Pi'$ such that, for all $u,s\in S$ and $v\in V$, $u\succ_{\pi'_s} v \Longrightarrow
  u\succ_{\pi_s} v$.
We need to show that if $S \in \widetilde{\calC}((V,\Pi'))$ then
  $S\in \widetilde{\calC}((V,\Pi))$.
Suppose this is not the case, then either
 (1)  $S\not\in \calC_{GS}((V,\Pi))$  or (2)
  $S\not\in \calC_{SA}((V,\Pi))$.
In Case (1), there exists $G \subset S$,  $G'\subset V-S$, $|G| =
  |G'|$, and
  bijections $(f_s: S\rightarrow G' | s\in S-G)$ such that $\forall s\in S-G, \forall
  u\in G$,  $u\prec_{\pi_s} f_s(u)$.
Then by the condition stated in {\sf Mon}, we have $u\prec_{\pi'_s} f_s(u)$,
  which shows $S\not\in \calC_{GS}(A')$.
In Case (2), there exists $G'\subset V-S$,
  bijections $(f_s: S\rightarrow G')$ such that $\forall s, u\in S$,
   $u\prec_{\pi_s} f_s(u)$.
Then by the condition stated in {\sf Mon}, we have  $u\prec_{\pi'_s} f_s(u)$,
  which implies that $S\not\in \calC_{GS}(A')$.

Suppose $\calC$ satisfies Axiom {\sf CRM}.
Consider  $\Pi,\Pi'$ as specified in Axiom {\sf CRM}.
{ Given $s\in S$,  the profiles $\pi_s$ and $\pi'_s$ are then assumed
to be identical on $V-S$, implying in particular that $\pi_s(V-S)=\pi'(V-S)$, and hence
 also that $\pi_s(S)=\pi'_s(S)$.  Furthermore, by the coherence assumption for members, there
exist $\sigma\in L(S)$}
such that $\forall u_1,u_2, s\in S$,
$u_1 \succ_{\pi'_s} u_2$ iff $ u_1 \succ_{\sigma} u_2$.
We need to show that if   $S\in \widetilde{\calC}((V,\Pi'))$  then
 $S \in \widetilde{\calC}(A)$.
Suppose this is not the case, then either
 (1)  $S\not\in \calC_{GS}(A)$  or (2)
  $S\not\in \calC_{SA}(A)$.

In Case (1),   there exists $G\subset S$, $G'\subset V-S$,
  $|G| = |G'|$, a set of
  bijections
  $(f_s: G\rightarrow G', s\in S-G)$, such that $\forall s\in S-G,
  u\in G$,
  $u\prec_{\pi_s} f_s(u)$.
Let $T \subset S$ be the set of $|G|$ least preferred
 elements  by $\sigma$.
We now show that there exists bijections
   $(f'_s: T\rightarrow G', s\in S)$ such that
 $\forall s\in S-T,   u\in T$,  $u\prec_{\pi'_s} f'_s(u)$,
   which would imply that $S\not\in \calC_{GS}((V,\Pi'))$.

Let us denote $T$ by $T=\{t_1,...,t_{|T|}\}$ such that
  $t_i \prec_{\sigma} t_{i+1}$.
Fix an $s\in S-T$, and let us denote
  $G$ by $G = \{g_1,...,g_{|T|}\}$ such that
  $g_i \prec_{\pi_s} g_{i+1}$, and
  denote $G'$ by $G' = \{g'_1,...,g'_{|T|}\}$ such that
  $g'_i \prec_{\pi_s} g'_{i+1}$.
By Proposition~\ref{prop:lex}, we then have that $g_i \prec_{\pi_s} g'_i$
 for all $i=1,\dots,|T|$.
In other words, $\pi_s(g_i) > \pi_s(g'_i)$.
We define $f_s'$ by mapping $t_i$ to $g'_i$.
Note that the positions of the preferences rankings
  of $S$ as a set are the same in $\pi'_s$ and $\pi_s$.
Because $T$ is the set of $|G|$ least preferred
 elements of $S$, we have  $\pi'_s(t_i) > \pi_s(g_i)$.
 Since $\pi'_s(g'_i) = \pi_s(g'_i)$
it then follows that
  $\pi'_s(t_i) > \pi_s(g_i) > \pi_s(g'_i) = \pi_s(g'_i)$.
Thus, $t_i \prec_{\pi'_s} g'_i$, and consequently,
  $S\not\in \calC_{GS}((V,\Pi'))$.
In Case (2),   there exists $G'\subset V-S$ and    a set of
  bijections
  $(f_s: S\rightarrow G', s\in S)$, such that $\forall s, u\in S$,
  $u\prec_{\pi_s} f_s(u)$.
By the similar argument as in Case (1) (by setting $T = S$), we can show that
  there exists bijections
   $(f'_s: S\rightarrow G', s\in S)$ such that
 $\forall s\in S,   u\in S$,  $u\prec_{\pi'_s} f'_s(u)$,
   which implies that $S\not\in \calC_{GS}((V,\Pi'))$.
Thus, $\widetilde{\calC}$ satisfies Axiom {\sf CRM}.

We can similarly prove that $\overline\calC$ satisfies
{\sf CRNM} if $\calC$ satisfies it.

Finally, by definition, $\calC\cap \calC_{GS}\cap\calC_{SA}$ satisfies
  {\sf GS} and {\sf SA}.
\end{proof}
	
\begin{Rule}{\sc (Comprehensive Community Rule)}
For a preference  network $A = (V,\Pi)$, a non-empty
 $S\subseteq V$ is a community according to $\calC_{comprehensive}$
   if and only if $S$
  satisfies both {\sf Group Stability} and {\sf Self-Approval} axioms.
In other words,
\[ \calC_{comprehensive} : = \calC_{GS}\cap\calC_{SA}.\]
\end{Rule}

We now prove that $\calC_{comprehensive}$ is indeed the most
  comprehesive community rule that satisfies all Axioms.

\begin{theorem}[Taxonomy: Lattice Top and Bottom]\label{th:taxonomy}
$\calC_{comprehensive}$ satisfies all Axioms.
Moreover, for any community function $\calC$ that satisfies all
Axioms, for every preference  network $A = (V,\Pi)$
\begin{eqnarray}\label{eqn:tax}
\calC_{clique}(A) \subseteq \calC(A) \subseteq \calC_{comprehensive}(A).
\end{eqnarray}
\end{theorem}
\begin{proof}	
$\calC_{all}(A) = 2^V{-\{\emptyset\}}$ satisfies Axioms {\sf A},
{\sf Mon}, {\sf CRM, CRNM}, {\sf WC} and {\sf Emb}.
Since
  $\calC_{comprehensive}= \calC_{all} \cap \calC_{GS}\cap \calC_{SA}$,
by the Intersection Lemma, 
  $\calC_{comprehensive}$ satisfies all Axioms.

{On the other hand, by Proposition~\ref{prop:Clique-Null}, any rule which satisfies {\sf WC} and {\sf Emb}, must satisfy the {\sf Cliques} Property, so} for any $\calC$ that satisfies all axioms,
  ${\calC_{clique}(A) \subseteq} \calC(A) \subseteq \calC_{GS}(A)\cap \calC_{SA}(A)$.
Thus
$\calC_{clique}(A) \subseteq \calC(A) \subseteq \calC_{comprehensive}(A).$
\end{proof}

Theorem \ref{th:taxonomy} shows that $\calC_{comprehensive}$ and
  $\calC_{clique}$ are the
  most inclusive and the most selective function,
  respectively, that satisfies all axioms.
While it is very easy to determine whether a subset in a preference
  network
  satisfies Property {\sf Clique},
  in  Section~\ref{Sec:complexity}
  we demonstrate that $\calC_{comprehensive}$ is highly
``non-constructive'' by showing that the decision problem for
  determining whether a subset in a preference network
  satisfies Axiom {\sf Self-Approval} or {\sf Group Stability}
  is {\sf coNP-complete}.

\subsection{The Lattice Structure of Community Rules}
\label{sec:lattice}

The Intersection Lemma provides us with
  a tool for exploring the taxonomy of community rules.
  In this subsection, we continue this exploration and make it more systematic using two lattice
  structures enjoyed by the community-rule taxonomy.

\begin{theorem}[Taxonomy: Lattice Structures of Community Rules]
\label{thm:lattice}
Let $\calbfC$ denote the family of all community rules that satisfies all eight axioms.
Let $\calbfC_B$ be a superset of $\calbfC$ that denotes the family of all
  community rules that
  satisfies Axioms {\sf A, Mon, CRM, CRNM, WC, Emb}.
\begin{enumerate}
\item  The algebraic structure $\calT = (\calbfC,\cup,\cap,\calC_{clique},\calC_{comprehensive})$ forms a bounded lattice with $\calC_{clique}$ as the lattice's bottom and $\calC_{comprehensive}$
  as the lattice's top.
\item  The algebraic structure
  $\calT_B = (\calbfC_B,\cup,\cap, \calC_{clique},\calC_{all})$ forms a bounded lattice
  with $\calC_{clique}$ as the lattice's bottom and $\calC_{all}$
  as the lattice's top.
\end{enumerate}
\end{theorem}

\begin{proof}
First, by definition, the two operations $\cap$ and $\cup$ over the community functions
  are both communitative and associative.
One can easily show that the two operations
  $\cap$ and $\cup$ satisfy the {\em absorption property}, that is,
   for any two $\calC_1,\calC_2\in \calbfC$
\begin{eqnarray*}
\calC_1 \cup (\calC_1\cap \calC_2) & = & \calC_1.\\
\calC_1 \cap (\calC_1\cup \calC_2) & = & \calC_1.
\end{eqnarray*}
For example, to see the first one, for any affinity network $A$, we have
\begin{eqnarray*}
(\calC_1 \cup (\calC_1\cap \calC_2))(A) = \calC_1(A) \cup (\calC_1\cap \calC_2)(A) =
 \calC_1(A) \cup (\calC_1(A)\cap \calC_2(A)) =  \calC_1(A).
\end{eqnarray*}

To complete the proof that $\calT$ and $\calT_B$  are lattices, we need to prove
  that  $\calT$ and $\calT_B$ are  closed under $\cap$ and $\cup$.
We organize the arguments as following:
\begin{itemize}
\item {\sf A, WC}: it is obvious that if $\calC_1$ and $\calC_2$ satisfies Axioms {\sf A} and ${\sf WC}$ then both $\calC_1\cup \calC_2$ and $\calC_1\cap \calC_2$ also satisfies Axioms {\sf A, WC}.
\item {\sf Mon, CRM, CRNM}:
Suppose $A= (V,\Pi)$, $A' = (V,\Pi')$, and $S\subset V$ are, respectively,
  two preference networks and a set considered in Axiom {\sf Mon}.
Then if $\calC_1$ and $\calC_2$ satisfy {\sf Mon}, we have
  $S\in \calC_i(A') \Rightarrow S\in C_i(A)$ for $i\in {1,2}$.
Thus, if $S\in \calC_1(A')\cap \calC_2(A')$ then $S\in \calC_1(A)\cap \calC_2(A)$, and
if $S\in \calC_1(A')\cup\calC_2(A')$ then $S\in \calC_1(A)\cup \calC_2(A)$.
Thus, both $\calC_1\cup \calC_2$ and $\calC_1\cap \calC_2$ also satisfy Axioms {\sf Mon}.
We can  argue analogously for Axioms {\sf CRM} and {\sf CRNM}.
\item {\sf Emb}: If both $\calC_1$ and $\calC_2$ satisfy {\sf Emb}, then
for any $A = (V,\Pi)$ and any ``embedded world'' $A' = ( V', \Pi')$
{such that $\Pi,\Pi'$ satisfy the assumption} of Axiom {\sf Emb}, we have
${\cal C}_i(A') = {\cal C}_i(A)\ \cap\ 2^{V'}$ for $i\in \{1,2\}$.
So
\begin{eqnarray*}
\calC_1(A')\cap \calC_2(A') &  = & \left(\calC_1(A) \cap\ 2^{V'}\right)\cap \left(\calC_2(A)\
\cap\ 2^{V'}\right)
= \left(\calC_1(A)\cap  \calC_2(A)\right) \cap\ 2^{V'} \\
\calC_1(A')\cup \calC_2(A') &  = & \left(\calC_1(A) \cap\ 2^{V'}\right)\cup \left(\calC_2(A)\
\cap\ 2^{V'}\right) = \left(\calC_1(A)\cup \calC_2(A)\right) \cap\ 2^{V'}.
\end{eqnarray*}
Thus, both $\calC_1\cup \calC_2$ and $\calC_1\cap \calC_2$ also satisfies Axioms {\sf Emb}.
\end{itemize}

Together, this shows that $\forall \calC_1,\calC_2\in \calbfC_B$, $\calC_1\cap \calC_2\in \calbfC_B  \mbox{ and } \calC_1\cup \calC_2 \in \calbfC_B.$
Thus, $\calT_B = (\calbfC_B,\cup,\cap, \calC_{clique},\calC_{all})$ is a lattice
  with $\calC_{all}$ as the lattice's top and $\calC_{clique}$ as the lattice's bottom
  { (where the
  former follows from the fact that $\calC_{all}(A) $ satisfies Axioms {\sf A},
{\sf Mon}, {\sf CRM, CRNM}, {\sf WC} and {\sf Emb}, while the
  latter follows from
  Proposition~\ref{prop:Clique-Null}).}
\begin{itemize}
\item {\sf GS, SA}:
Assume  $\calC_1 \in \calbfC$ and $\calC_2 \in \calbfC$ satisfy Axioms {\sf GS} and {\sf SA}.
 We can then argue as
for Axiom {\sf Mon} above to show that
 both $\calC_1\cup \calC_2$ and $\calC_1\cap \calC_2$ satisfies Axioms {\sf GS, SA}.
\end{itemize}
Thus, $\calT = (\calbfC,\cup,\cap)$ is a lattice.
By Theorem \ref{th:taxonomy},
  $\calC_{comprehensive}$ is the lattice's top and $\calC_{clique}$ as
  the lattice's bottom of $\calT$.
\end{proof}

Theorem~\ref{thm:lattice}
 allows us to have a notion of the closure of an arbitrary community rule with respect to these six axioms.
 In order to define it, we say that a community rule $\calC_2$ {\emph contains} a rule $\calC_1$ if
 $\calC_1(A)\subset \calC_2(A)$ for all preference networks $A$.

\begin{theorem}
\label{thm:unique}
	Given a community rule $\calC$, there exists a unique smallest community rule, denoted $\overbar{\calC}$, that contains $\calC$ and satisfies all community axioms besides {\sf SA} and {\sf GS}.
\end{theorem}

\begin{proof}
	Consider the set $\widehat\calbfC$ of all community rules that contain $\calC$ and satisfy these six axioms.  Note that it is non-empty because $\calC_{all}$ is guaranteed to contain $\calC$.  Apart from some technical issues to be addressed below,  if we take the intersection of all the communities in this set,
the resulting rule $\overbar{\calC}$ will still satisfy all six axioms by the proof of Theorem~\ref{thm:lattice}, and thus be the smallest community rule of the set.

The technical issues to which we alluded  above stem from the fact that, in general, the
set $\widehat\calbfC$ contains uncountably many community rules.  The community rule $\overbar{\calC}$ is thus defined
by an uncountable intersection, while Theorem~\ref{thm:lattice} {\it a priori} only allows one to argue about countably many intersections.  But it turns out that while $\widehat\calbfC$ is uncountable, when checking the axioms, one never has to consider more than a finite set of rules, allowing one to apply the reasoning from the proof of Theorem~\ref{thm:lattice}
to show that $\overbar{\calC}$ does satisfy all desired axioms.

To make this precise, we recall that a community rule is given by a sequence of functions,
$\calC_V: (V,\Pi)\mapsto \calC_V(V,\Pi)\subset 2^{V}{-\emptyset}$, where $V$ runs over the non-empty finite subsets of countable reference set $V_0$.
Expressing both $\overbar{\calC}$ and the rules in $\calC'\in\widehat\calbfC$ as sequences,
$\calC=(\calC_V)$ and ${\calC'}=(\calC_V')$, we have
\[
\overbar{\calC}_V((V,\Pi))=\bigcap_{\calC'\in{\widehat\calbfC}}\calC'_V((V,\Pi)).
\]
Hoverer, when verifying the six axioms for $\overbar{\calC}$, we only have to deal with a given finite set $V$ at a time (or, in the case of Axiom {\sf Emb}, all subsets $V'\subset V$ of a finite set $V$); and for a finite set $V$,
$\overbar{\calC}_V$ can be expressed as the intersection over a finite subset of $\widehat\calbfC$, which means when checking the axioms for $\overbar{\calC}_V$, we can use Theorem~\ref{thm:lattice}.
\end{proof}

The Intersection lemma serves as a bridge between the two lattices from Theorem~\ref{thm:lattice}: We can obtain
 the lattice
$\calT = (\calbfC,\cup,\cap,\calC_{clique},\calC_{comprehensive})$
from the lattice $\calT_B = (\calbfC_B,\cup,\cap, \calC_{clique},\calC_{all})$
  by intersecting the community functions on the lattice points
 of $\calT_B$
  with $\calC_{GS}\cap \calC_{SA}$,  followed by merging the lattices points with
  identical community functions.
By moving the intersection up the lattice $\calT_B$, we can define more inclusive community
 rules that satisfy all eight axioms.
For example, by intersecting the lattice top ($\calC_{all}$) of $\calT_B$ with
  $\calC_{GS}\cap \calC_{SA}$, we obtain the
  lattice top ($\calC_{comprehensive}$) of $\calT$.
\begin{remark}
	Note that Theorem~\ref{thm:unique} and the Intersection Lemma give us a reasonable mapping from arbitrary community rules to community rules that satisfy all our axioms.  Namely, for a given community rule, $\calC$, first take the unique smallest community rule that contains $\calC$ and satisfies all axioms besides {\sf SA} and {\sf GS} (as in Theorem~\ref{thm:unique}), then apply the intersection from the Intersection Lemma.  The mapping can therefore be formulated as
	$$\calC \longmapsto \overbar{\calC} {\cap} \calC_{SA} {\cap}\calC_{GS}.$$
\end{remark}

As an example, consider the community rule
{$\calC_{1}$} that admits all 
singletons (i.e., subsets of size 1)
as communities  and nothing else.  Because
{$\calC_{1}$}
 only violates {\sf WC} of the axioms besides {\sf SA},
{$\overbar{\calC_{1}}$}
in addition to all  singletons
also contains all cliques (thanks to the influence of {\sf Emb}).  From this, all communities that don't satisfy {\sf SA} are removed:
 i.e., 
all singletons that do not rank themselves first.
As the reader may have already guessed, what remains happens to be the Clique Rule.  In other words,
$$
{\overbar{\calC_{1}}}{\cap} \calC_{SA} {\cap \calC_{GS}} = \calC_{clique}$$

In a small step up the lattice $\calT_B$ from the Clique Rule,
  we consider the following community function.

\begin{Rule}[Relaxed Clique Rule]
 For a non-negative function $g:{\mathbb N}\to{\mathbb N}\cup\{0\}$,
 a non-empty subset
    $S\subseteq V$ is a community in $A = (V,\Pi)$
    if and only if $\forall u,s\in S$,
  $\pi_s(u) \in [1:|S| + g(|S|)]$.
We denote this community function by $\calC_{Clique(g)}$.
 \end{Rule}
\begin{proposition}
\label{prop:rel-clique-rule}
$\calC_{Clique(g)}\in \calbfC_B$ and hence $(\calC_{Clique(g)}\cap \calC_{GS}\cap \calC_{SA})$ satisfies all eight axioms.
\end{proposition}

{
\begin{proof}
The (straightforward) proof is left to the reader.
\end{proof}
}

We will show in Section~\ref{Sec:complexity} below  that $\calC_{comprehensive}$ is highly
``non-constructive'' by proving that the decision problem for
  determining whether a subset in a preference network
  satisfies Axiom {\sf Self-Approval} or {\sf Group Stability}
  is {\sf coNP-complete}.  On the other hand, we will see that
the community rule given by $(\calC_{Clique(g)}\cap \calC_{GS}\cap \calC_{SA})$
  can be constructive if $g$ is small, see Proposition~\ref{prop:g-CO-NP} in Section~\ref{Sec:complexity} below.

As $g$ varies from $0$ to $\infty$, the community function
  $\calC_{Clique(g)}$ moves up the lattice $\calT_B$ from
  $\calC_{clique}$ to $\calC_{all}$.
The intersection with
  $\calC_{SA}\cap \calC_{GS}$ provides us a ``vertical'' glimpse of
  the taxonomy lattice $\calT$.
In particular, as the
  community rules along this vertical path become more
  inclusive (when $g$ increases), they become less
  constructive for community identification.
An alternative ``vertical'' glimpse can be gained by
  following ``harmonious-path'' in the lattice $\calT_B$ for community
  rules formulated by pairwise comparisons.

\begin{Rule}[Harmonious Path]\label{rule:lambda-harm}
For $\lambda \in [0: 1]$, a non-empty subset $S$ is a
{\em $\lambda$-harmonious community} in $A  = (V,\Pi)$
if $\forall u\in S, v \in V -S$, at least $\lambda$-fraction
  of  $\{\pi_s\ : \ s\in S\}$  prefer $u$ over $v$.
We denote this community function by $\calH_{\lambda}$.
 \end{Rule}

Using the similar argument as in Theorem \ref{theo:H}, we can prove
\begin{proposition}
For all $\lambda \in [0: 1]$, $\calH_{\lambda} \in \calbfC_{B}$.
Thus,
$\calH_{\lambda} \cap \calC_{GS}\cap \calC_{SA}$ satisfies all eight axioms,
 $\forall \lambda \in [0: 1]$.  Further, for $\lambda \in (1/2: 1]$,
  $\calH_{\lambda}$ satisfies Axiom {\sf SA}, and therefore all axioms but
  Axiom {\sf GS}.
\end{proposition}

\begin{proof}
The (easy) proof is left to the reader.
\end{proof}

Therefore, as $\lambda$ varies from $1$ to $0$, the community function
  $\calH_{\lambda}$ moves up the lattice $\calT_B$ from
  $\calH_{1} = \calC_{clique}$ to $\calH_{0} = \calC_{all}$, and so does its non-constructiveness,
  see Proposition~\ref{prop:H-complexity} in Section~\ref{Sec:complexity}.

\section{Complexity of Community Rules}\label{Sec:complexity}

\subsection{Complexity of determining {\sf Group Stability} and {\sf Self Approval}}
  In this section,
  we demonstrate that $\calC_{comprehensive}$ is highly
``non-constructive'' by showing that the decision problem for
  determining whether a subset in a preference network
  satisfies Axiom {\sf Self-Approval} or {\sf Group Stability}
  is {\sf coNP-complete}.
Our reduction also provides examples of preference networks derived from
  {\sf 3-SAT} instances.

\begin{theorem}\label{thm:SA}
	It is {\sf coNP-complete} to determine whether a subset $S \subset V$
is {\sf self-approving}
in a given preference network $A = (V, \Pi)$.
\end{theorem}

 Before starting the proof, we introduce a notation which we will
use throughout this section.  Given a preference profile $(V,\Pi)$ and a
non-empty
set $S\subset V$, we say that a set $G'\subset V-S$ is a \emph{ witness that
 $S$  is not  self-approving}, if $S$ lexicographically prefers $G'$ to $S$,
 and we say that a pair $(G,G')\subset S\times (V-S)$ is a \emph{ witness that
 $S$  is not group-stable} if $S-G$ lexicographically prefers $G'$ to $G$.  Finally,
 we say that $G\subset S$ \emph{ threatens the stability of $S$}
 if there exists a $G'\subset V-S$ such that $S-G$ lexicographically prefers $G'$ to $G$.

\begin{proof}
	We reduce {\sf 3-SAT} to this decision problem: Suppose
	$\mathbf{c} = (c_1, \ldots, c_m)$ is a {\sf 3-SAT} instance
	with Boolean variables $\mathbf{x} = (x_1, \ldots, x_n)$
	(i.e., $c_j = \{u_j, v_j, w_j\} \subset\cup_{i=1}^n \{x_i,\bar{x}_i\}$
).
 We define a preference network as follows:
\begin{itemize}
\item $V = A \cup B \cup D \cup X$ has $m+n + m+2n$ members, where $A = \{a_1, \ldots, a_m\}$, $B = \{b_1, \ldots, b_n\}$, $D = \{d_1, \ldots, d_m\}$, and $X = \{x_1, \ldots, x_n, \bar{x}_1, \ldots, \bar{x}_n\}$.  The distinguished subset will be $S = A \cup B$, and for convenience we will denote its complement as $U = D \cup X$.
\item Since we will focus on subset $S$, here we only define the preferences of members in $S$.  The preferences of $U$ can be chosen arbitrarily.
\begin{itemize}
\item Member $b_i$ has preference $D \succ A \succ \{x_i, \bar{x}_i\} \succ \{b_i\} \succ X - \{x_i, \bar{x}_i\} \succ B - \{b_i\}$, where preferences between elements of each set can be chosen arbitrarily.
\item Member $a_j$ has preference $c_j \succ \{a_j\} \succ D \cup X-c_j \succ B \cup A - \{a_j\}$, where again preferences between elements of each set are arbitrary.
\end{itemize}
\end{itemize}

	Intuitively, members of $A$ are used to enforce clause consistency (i.e., make sure each clause is satisfied) and members of $B$ are used to enforce variable consistency (no variable to both true and false at the same time).  Subsets of $X$ naturally constitute an assignment of the variables, and $D$ provides necessary padding in order to apply {\sf Self-Approval}.
	
	We now show that $S$  is not  {\sf self-approving}
   if and only if the {\sf 3-SAT} instance is satisfiable.
	
In one direction, suppose $Y = \{y_1, \ldots, y_n\}$ where $y_i \in \{x_1, \bar{x}_i\}$ is a satisfying assignment for the {\sf 3-SAT} instance.  Let $G' = Y \cup D$.  Now consider the bijection, $f$, where $f(a_j) = d_j$ and $f(b_i) = y_i$.  It is not hard to see that for all $s \in S$ and all $i$, $f(s) \succ_{\pi{b_i}} s$.  All that is left is to find similar bijections for each $a_j$.  First, note that for $a_j$ all bijections $f_j$ trivially satisfy $f_j(s) \succ_{\pi{a_j}} s$ where $s \in B \cup A - \{a_j\}$, since this set is ranked at the bottom of $\pi_{a_j}$.  Therefore it is sufficient to show that there exists an element of $G'$ that $a_j$ prefers to itself.  This happens so long as one of the literals from its clause is in $G'$, which must be true by the fact that $Y$ is a satisfying assignment.
	
In the other direction, suppose $G' \subset U= D \cup X$ is a witness
	that $S$ {\sf is not self-approving}.  We note
	the following:

\begin{itemize}
\item $D \subset G'$ otherwise any $b_i$ will have a member of $A$ that cannot be mapped to a more preferred member of $G'$.
\item Let $Y = X \cap G'$.  Then $|Y| = n$ by the above fact and the fact that $|G'|=n+m$.
\item $\{x_i, \bar{x}_i\} \cap G' \neq \emptyset$ by $b_i$'s preference, and by the pigeonhole principle the literals of $Y$ are consistent (i.e. $\{x_i, \bar{x}_i\} \nsubseteq Y$).
\item $c_j \cap Y \neq \emptyset$ by $a_j$'s preferences.
\end{itemize}

	Therefore the variable assignment implied by $Y$ is a satisfying assignment for the {\sf 3-SAT} instance.
\end{proof}

The following ``padding'' lemma allows us to reduce various complexity results concerning community axioms to Theorem~\ref{thm:SA}.
\begin{lemma}\label{lem:padding}
Let ${\emptyset\neq} S\subset V\subset V'$ be such that
the size of $\tilde S=V'-V$ is at least $|S|$, and let $S'=S\cup \tilde S$.  Then each preference profile $\Pi$ on $V$ can
be mapped onto a preference profile $\Pi'$ on $V'$ such that
\begin{enumerate}
\item[(i)] $S'\in \calC_{GS}(V',\Pi')\cap  \calC_{SA}(V',\Pi')\Leftrightarrow S'\in \calC_{GS}(V',\Pi')$.
\item[(ii)] $S'\in \calC_{GS}(V',\Pi')\Leftrightarrow S\in \calC_{SA}(V,\Pi)$.

\end{enumerate}
\end{lemma}
\begin{proof}
Since $|\tilde S|\geq |S|$, we can find a surjective map $g:\tilde S\to S$.  Given such a map, define $\Pi'$ arbitrarily,
except for the following two constraints:
\begin{itemize}
\item If $s \in S$, then $\pi'_s$ ranks all of  $S'=\tilde S\cup S$ before anyone in $V-S=V'-S'$;
\item If $\tilde{s} \in \tilde{S}$, then $\pi'_{\tilde s}$ ranks all of $\tilde{S}$ first, and then gives the rank
$\pi'_{\tilde s}(v)=|\tilde S|+\pi_{g(\tilde s)}(v)$ to every $v\in V=V'-\tilde S$.
\end{itemize}
Since every $s\in S\subset S'$ ranks all of $S'$ before $V'-S'$, no subset $G'\subset V'-S'$ can be lexicographically preferred by $\pi'_s$ to
a subset of $S'$.  As a consequence, $S'$ is trivially self-approving with respect to $\Pi'$, proving  statement (i).

Furthermore,
$G$ cannot threaten the stability of $S'$ if $G\subset S'$ is such that $(S'-G)\cap S\neq \emptyset$.  If $G\subset S'$
threatens the stability of $S'$, we therefore must have that $G\supset S$.  On the other hand, if $G\supsetneq S$, then
$G$ contains an element $\tilde s\in\tilde S$ which means that no set $G'$ can be lexicographically preferred $G$,
since all elements of $S'$ prefer all of $\tilde S$ to anyone in $V'-S'$.

Thus $G$ can only threaten the stability of $S'$ if $G=S$.
In other words, $S\notin \calC_{GS}(\Pi')$ if and only if there exists $G'\subset V'-S'$
such that for all $\tilde s\in \tilde S=S'-G$, $G'$ is lexicographically preferred to $S$ with respect to
$\pi'_{\tilde s}=\pi_{g(\tilde s)}$.  Since by assumption, the image of $\tilde S$ under $g$ is all of $S$,
this is equivalent to the statement that for all $s\in S$, $G'$ is lexicographically preferred to $S$ with respect to
$\pi_{s}$, which is the condition that $G'$ is a witness to $S\notin \calC_{SA}(\Pi)$, proving statement (ii).

\end{proof}

Given this lemma, the next two theorems are immediate corollaries to Theorem~\ref{thm:SA}.

\begin{theorem}\label{thm:GS}
	It is {\sf coNP-complete} to determine whether a subset $S \subset V$
is {\sf group-stable}
in a given preference network $A = (V, \Pi)$.
\end{theorem}

\begin{theorem}\label{thm:SAorGS}
	It is {\sf coNP-complete} to determine whether a subset $S \subset V$
is a member of
$\calC_{comprehensive} = \calC_{GS}\cap \calC_{SA}$
for a given preference network $A = (V, \Pi)$.
\end{theorem}

\subsection{Complexity of the rules  $\calC_{Clique(g)}$ and $\calH_\lambda$}

We now prove although testing membership for
  $\calC_{comprehensive} = \calC_{GS}\cap \calC_{SA}$ is co-NP complete,
the community rule given by $(\calC_{Clique(g)}\cap \calC_{GS}\cap \calC_{SA})$
  can be constructive if $g$ is small.

\begin{proposition}
\label{prop:g-CO-NP}
Given a preference network $A = (V,\Pi)$ and a subset $S\subseteq V$,
then we can determine in $O(2^g |S|^{g+3} )$  time whether or not $S\in (\calC_{Clique(g)}\cap \calC_{GS}\cap \calC_{SA})(A)$.
Particularly, if $g = \Theta(1)$, then this decision problem is in P.
However, the decision problem is co-NP complete for $g = |S|^{\delta}$ for any constant $\delta \in (0,1]$.
\end{proposition}
\begin{proof}
It takes time $O(|S|^2)$ to check whether $S \in  \calC_{Clique(g)} $.

Next we show that it takes time $O(|S|^3 2^g) $ to check if $S \in \calC_{SA}(A)$.
Indeed, suppose $G'\subseteq V-S$ is a witness that $S  \notin \calC_{SA}(A)$.
We claim that this implies that
  $G'\subset\pi_s^{-1}([1:|S| + g])$  $\forall s\in S$.
Suppose this is not true for some $s \in S$.  Then  $\exists v\in G'$ such that
$\pi_s(v)>|S|+g$, which in turn implies that
$v\prec_{\pi_s} u \,\forall u\in G$ as $\pi_s (u) \in [1:|S| + g]$ $\forall u\in G$.  Thus
  there exists no bijection $f_s: S\rightarrow G'$ with the property
  $f^{-1}_s(v) \prec_{\pi_s} v$, contradicting
  the assumption that  $G'\subseteq V-S$ is a witness that $S \not\in \calC_{GS}(A)$.
We can thus identify the set of all witnesses as follows:
(1) Choose $s\in S$, and let $T_s = \pi_s^{-1}([1:|S| + g]) - S$.
(2) Choose a subset  $G' \subseteq T_s$.
(3) Test if $G'$ is a witness that $S \not\in \calC_{SA}(A)$.
First note that we are dealing with at most $|S|2^g $ subsets.
By Proposition \ref{prop:lex}, we can conduct the test of Step 3 performing $|S|$ integer sorting.
Thus, the total complexity for Steps 1-3 is $O(|S|^3 2^g)$.

We can similarly test for group stability for $S\in \calC_{Clique(g)}$.
Suppose $(G, G')$ is a witness that
  $S \not\in \calC_{GS}(A)$.
Then, it must be the case that
$G'\subset \pi_s^{-1}([1:|S| + g])$ $\forall s\in S-G$.
Suppose this is not true for some $s \in S-G$.
  Then there must be a $v\in G'$
  such that $v\prec_{\pi_s} u, \forall u\in G$ as $u \in \pi_s^{-1}([1:|S| + g])$, which implies that
  there exists no bijection $f_s: G\rightarrow G'$ with the property
  $f^{-1}(v) \prec_{\pi_s} v$, contradicting
  the assumption that $(G \subset S, G'\subseteq V-S)$ is a witness that
  $S \not\in \calC_{GS}(A)$.

We say $G'$ is a {\em potential witness} to  $S \not\in \calC_{GS}(A)$
 if there exists $G\subset S$, such that  $(G, G')$ is a witness to
  $S \not\in \calC_{GS}(A)$.
We can identify the set of all potential witnesses as follows:
(1) Choose $s\in S$, and let $T_s = \pi_s^{-1}([1:|S| + g]) - S$.
(2) Choose a subset of $G' \subseteq T_s$.
(3) Test if $G'$ is a potential witness.
Again, we are dealing with at most $|S|2^g $ subsets.
As there are at most $|S|^{|G'|} \leq |S|^g$ candidates $G$ to test for (using Proposition \ref{prop:lex}),
Steps 1-3 takes at most  $O(2^g |S|^{g+3} )$ time.

To show that for large $g$ the problem of determining whether or not a set lies
 in $\calC_{Clique(g)}\cap \calC_{GS}\cap \calC_{SA}$ is in co-NP, we
 reduce the problem to the one of determining whether for a given
 preference network $A=(V,\Pi)$,  a set $S\subset V$ is a member of
 $(\calC_{GS}\cap\calC_{SA})(A)$.  To define the reduction,
 we enlarge both $V$ and $S$ by a large, disjoint set $\tilde S$:
$V=V'\cup\tilde S$, $S'=S\cup\tilde S$, where $\tilde S$ is chosen large enough to guarantee that $g(|\tilde S|)\geq |V|$, implying in particular
that $|S'|+g(|S'|)\geq |S'|+|V|\geq |\tilde S|+|V|=|V'|$.  Due to this fact, we have that $S'\in \calC_{Clique(g)}(V',\Pi')$ for all preference
profiles $\Pi'$ on $V'$.
The statement now follows with the help of Lemma~\ref{lem:padding} and Theorem~\ref{thm:SA}.
\end{proof}

Our final proposition  in this subsection concerns the complexity of determining whether a set $S$
lies in the class $\calH_{\lambda}$.

\begin{proposition}
\label{prop:H-complexity}
Given a preference network $A = (V,\Pi)$ and a subset $S\subseteq V$,
we can determine in polynomial time whether
  $S\in (\calH_{\lambda}\cap \calC_{GS}\cap \calC_{SA})(A)$
if $(1-\lambda)|S| < 2$, while it is co-NP complete to
  answer this question  if $(1-\lambda)|S|\geq 16$.
\end{proposition}

\begin{proof}
We start with the proof of the positive statement.
To this end, we first note that it takes
  $(|V|-|S|) |S|^2=O(|V|^3)$ comparisions
  to check whether $S\in \calH_{\lambda}$.

Next we show that if $S\in \calH_{\lambda}$, then the only groups $G\subset S$
that can threaten the stability of $S$ are those for which
\[
|S-G|\leq 2\lfloor(1-\lambda)|S|\rfloor-1.
\]
Indeed, assume that
$(G,G')$ is a witness for $S\notin\calC_{GS}(\Pi)$, and
let $g=|G|$.
The assumption that  $S\in \calH_{\lambda}$ then implies that for all
$(u,v)\in G\times G'\subset S\times (V-S)$, there
are at most
\[
m=|S|-\lceil \lambda|S|\rceil=
\lfloor (1-\lambda)|S|\rfloor
\]
elements
$s\in S-G\subset S$ such that $v\succ_{\pi_s} u$. Thus the sum
over all triples $(u,v,s)\in G\times G'\times (S-G)$ obeying this
condition can be at most $g^2 m$.  On the other hand,
if $s$ lexicographically prefers $G'$ over $G$, the number of
pairs $(u,v)\in G\times G')$ obeying the above condition is at
least $\frac{g(g+1)}2$, given a lower bound of $|S-G|\frac{g(g+1)}2$ on the
above number of triples.  This proves that
$
|S-G|\leq\frac{2g}{g+1}m,
$
and hence $|S-G| \leq 2m-1$, where in the last step we used that both $|S-G|$ and $m$ are integers.

Thus for $(1-\lambda)|S|<2$, we
 may assume that  $G-S$ has size $1$
(the case $G-S=\emptyset$ is trivial), which shows that
   there are at most $|S|$ possible choices for $G$.
Given $G$,
  we then only have to check whether a potential
  $G'\subset V-S$ is lexicographically preferred to $G$
  by a single linear order $\pi_s$,
  where $s$ is the single element of $S-G$.
Using Proposition~\ref{prop:lex},
  the existence of such a $G'$ can be checked by greedily
  choosing the first $|G|$ elements of $V-S$ with respect
  to $\pi_s$.
If this set is lexicographically preferred to $G$, we know that
 $S\notin\calC_{GS}$, and if for all $G$
  considered in the first step, the greedily found $G'$
  is not lexicographically preferred to $G$, $S\in\calC_{GS}$.
Since all $S\in\calH_\lambda$ are self-approving when $\lambda>1/2$,
this completes the proof of the positive statement.

To prove the negative statement, we use that it is NP-complete
  to determine whether in a formula consisting of $3$-clauses, every
  clause is satisfied by exactly one literal in the clause, and that
  this problem stays NP complete if we restrict ourselves to the case
  where each variable appears in exactly $3$ clauses (cubic 1-in-3
  SAT) \cite{moore:hardtiles}.
Note that this means that we can
  partition the set of clauses into $k=7$ classes such that the
  clauses in each class don't share any variables (to see this,
  consider the graph obtained by joining two clauses whenever they
  share a variable; this graph has maximal degree at most $6$, and
  hence can be colored by $7$ colors, given the desired partition).

Thus consider $n$ boolean variables $\{x_1,\dots,x_n\}$ and $k$ sets
of 3-in-1 clauses $C_i$ such that the clauses in each $C_i$ have no
common variables.  We define $X$ as the set of literals, $X=\{x_1,\bar
x_1,\dots, x_n,\bar x_n\}$, and choose two additional sets $Y$ and
$T$, of size $n$ and $2k+2$, respectively.  It will be convenient to
label the elements of $Y$ as $y_1,\dots, y_n$, and the elements of $T$
as $1,\dots, 2k+2$.  Set
\[
S=Y\cup T\quad\text{and}\quad V=S\cup X,
\]
and choose $\Pi'\in L(V)^V$ of the form
\begin{itemize}
\item  If $s\in Y$, $\pi'_s$ ranks all of $T$ first, followed by everyone in $Y$, followed by everyone in $X$
\item If $s \in T$, $\pi'_s$ ranks all of $T$ first, then ranks $V-T$ according to a yet to be determined $\pi_s\in L(Y\cup X)$.
\item If $v \in V-S$, $\pi'_v$ is arbitrary.
\end{itemize}
With this ranking, everyone in $Y$ ranks all of $S$ above all of $V$,
showing that $S\in \calH_\lambda((V,\Pi'))$ as long as $|S-Y|\leq
(1-\gamma)|S|$, i.e., as long as $(1-\gamma)|S|\geq 2(k+1)=c_2$.  It
also shows that $S$ is self-approving, since everyone in $S$ prefers
all of $T$ to all of $V-S$, which does not allow for a subset
$G'\subset V-S$ such that $G'$ is lexicographically preferred to $S$
by everyone in $S$.  By the same reasoning we also see that $S$ is
group-stable against any subgroup $G$ which has a non-zero
intersection with $T$.  Finally, $S$ is also stable against any subgroup
$G$ such that $Y\setminus G\neq\emptyset$, since for such a subset,
$S-G$ contains an element $s\in Y$ which prefers everyone in $S$ to
everyone outside $S$.

Thus the only subgroup $G$ against which $S$ could be unstable is the
set $G=Y$, i.e.,
$S\notin (\calH_{\lambda}\cap \calC_{GS}\cap \calC_{SA})(V,\Pi')$
  if and only if there exists a subset $G'\subset
X$ such that $G'$ is lexicographically preferred to $G=Y$ by everyone
in $T$.  We now show that by defining $\Pi$ appropriately, such a $G'$
exists if and only if the the 1-in-3 SAT problem given by
$C_1,\dots,C_k$ has a satisfying assignment.

We first define $\pi_1$ and $\pi_2$:
\[
\begin{aligned}
\pi_1&=[x_1,\bar x_1,y_1,\dots, x_n,\bar x_n,y_n]
\\
\pi_2&=[x_n,\bar x_n,y_n,\dots, x_1,\bar x_1,y_1].
\end{aligned}
\]
Clearly, $G'$ is lexicographically preferred to $G$ by both $\pi_1$
and $\pi_2$ if $G'$ contains exactly one of $x_i$ and $\bar x_i$ for
each $i$.  On the other hand, if $G'$ is lexicographically preferred
to $G$ by $\pi_1$, then by Proposition~\ref{prop:lex}, $G'$ must
contain at least one of $x_1$ and $\bar x_1$, and if it is
lexicographically preferred to $G$ by $\pi_2$, it can contain at most
one of $x_1$ and $\bar x_1$.  Continuing by induction, we see that
$G'$ is lexicographically preferred to $G$ by both $\pi_1$ and $\pi_2$
if and only if $G'$ contains exactly one of $x_i,\bar x_i$ for all
$i$, i.e., if $G'$ corresponds to a truth assignment for the variables
$x_1,\dots,x_n$.

In a similar way, if $C_i$ consists of the clauses
$\{z_1,z_2,z_3\},\{z_4,z_5,z_6\},\dots, \{z_{3\ell-2},z_{3\ell-1},z_\ell\}\subset X$, we define
\[
\begin{aligned}
\pi_{2i+1}&=[z_1,z_2,z_3,y_1,z_4,z_5,z_6,y_2,\dots, z_{3\ell-2},z_{3\ell_1},z_\ell,y_\ell,Q]
\\
\pi_{2i+2}&=[Q,z_{3\ell-2},z_{3\ell_1},z_\ell,y_\ell,\dots,z_4,z_5,z_6,y_2,z_1,z_2,z_3,y_1],
\end{aligned}
\]
where $Q$ ranks everyone in $X-\{z_1,\dots,z_{3\ell}\}$ before the
remaining elements $y_{\ell+1},\dots,y_n\in Y$.  Now the first ranking
enforces that at least one literal of the clause $\{z_1,z_2,z_3\}$ is
chosen, while the last enforces that there is at most one such
literal.  Combining these two and continuing by induction, we see that
$G'$ is lexicographically preferred to $G$ by both $\pi_{2i+1}$ and
$\pi_{2i+2}$ if and only if exactly one literal of each clause in
$C_i$ is chosen.

Putting everything together, we see that the 3-in-1 SAT problem has a
satisfying assignment if and only if $S\notin (\calH_{\lambda}\cap
\calC_{GS}\cap \calC_{SA})(V,\Pi')$.
\end{proof}

\subsection{Number of Potential Communities}

\begin{proposition}
\label{prop:comprehensive-size}
Assume that $n\geq 8$.
There exists a preference network $A= (V,\Pi)$ such that
  $\calC_{comprehensive}(A) \geq 2^{n/2}$.
\end{proposition}

\begin{proof}
The preference profile, $\Pi_{H\&S}$, that is about to be described
has been dubbed the ``hero and sidekick'' example as will soon become
clear.  Consider a world composed of $n/2$ hero-sidekick duos.  Each
member of a hero-sidekick duo first prefers the hero of that
duo then  the sidekick of the duo, then all other heroes,
followed lastly by all other sidekicks (in some fixed but arbitrary
order).  Now consider a subset, $S$, that is composed of all heroes
and an arbitrary set of sidekicks.  Note that because there are
$2^{n/2}$ different sets of sidekicks, it is sufficient to show that
$S$ is a community in $\calC_{comprehensive}([n],\Pi_{H\&S})$.

First, note that $S$ clearly satisfies {\sf SA}.

To show that $S$ satisfies {\sf GS}, consider two sets $G\subset
  S$ and $G'\subset V-S$ of equal size. We first note that it will be
  enough to consider the case where $(S-G)\times G$ contains no
  hero-sidekick pair $(u,v)$, since otherwise $u$ would prefer $v$
  over everyone else, in particular over everyone in $G'$.  Applying
  this to the sidekicks in $G$, we conclude that $G$ must contain at
  least as many heros as sidekicks.  On the other hand, $G'$ can't be
  lexicographically preferred to $G$ if $G$ contains at least two
  heros, showing that only two cases are possible: $G$ consisting of a
  hero-sidekick pair, or $G$ made up of just a single hero.  But
  neither one leads to a counter example if $|S-G|>|G|=|G'|$, since
  then we can find an $s\in S-G$ which is not the partner of any
  sidekick in $G'$, which means that $s$ prefers the hero in $G$ to
  everyone in $G'$.  Since $S$ contains all heros by assumption, we
  see that $S$ is group stable as soon as $n\geq 8$.
\end{proof}

\section{Stability of Communities}\label{Sec:Stability}
In this section, we  consider several stability measures
  and their impact on community structures.
In particular, we focus on preference perturbations in Section
\ref{sec:perturbation}
  and the concept of stable fixed points of an aggregation function in
  Section \ref{sec:StableFixedPoints}.
In both subsections, we will use
 \B3CT self-determined communities as our main examples
  to illustrate these measures.
In Section \ref{subsec:propHarmonious}, we will study the
  structure of stable harmonious communities.

\subsection{Community Stability with Respect to Preference  Perturbations}
\label{sec:perturbation}

We first study the structure of self-determined
  communities that remain self-determined even after a certain degree
  of changes in their members' preferences.

\begin{definition}[Preference Perturbations]
Let ${\emptyset\neq} S\subseteq V$, and let $\Pi$, $\Pi'$ be two preference profiles over $V$.
For $0 \leq \delta\leq 1$, we say $\Pi'$ is a
   $\delta$-{\em perturbation of $\Pi$ with respect to $S$}
 if
   $$\max_{v\in V}\left|\left\{i\in S :\ \pi_i(v)\neq \pi_i'(v) \right\}\right|\leq \delta |S|.$$

Given any community rule $\calC$ and a preference network $A=(V,\Pi)$,
   we say that a community $S\in \calC(A)$ is
   {\em stable under $\delta$-perturbations}
   if $S\in \calC((V,\Pi'))$ for all $\Pi'$ that are
   $\delta$-perturbation of $\Pi$ with respect to $S$.
\end{definition}

In other words, a preference profile is a
   $\delta$-perturbation of another profile if, for each $v\in V$, at most a $\delta$-fraction of
   the members of $S$  changed their preference of $v$.  Recalling Definition~\ref{def:a-b-community}, we now
   state our first stability result for  \B3CT self-determined communities.

\begin{proposition}
For any preference network $A = (V,\Pi)$,
if $S\subset V$
  is a \B3CT self-determined community that is stable
  under $\delta$-perturbations, then
$\exists \alpha  >\delta$ such that $S$ is an
$(\alpha,\alpha-\delta)$-\B3CT community.
Conversely, if  $S\subset V$ is an
  $(\alpha,\beta)$-\B3CT self-determined community,
  then it is stable under $(\alpha-\beta)/2$-perturbations.
\end{proposition}
\begin{proof}
Let $u^* = \mbox{argmin}\{\phi_S^\Pi(u): u\in S\}$, and let $\alpha =
\phi_S^\Pi(u^*)/|S|$.
We now prove that the condition of the
  proposition implies
$\alpha  >\delta$.
Suppose this is not true. Letting
  $T =\{s\in S : \pi_s(u^*) \leq |S|\}$,
  we have $|T| = \alpha |S| \leq\delta |S|$.
Now consider a preference profile $\Pi'$ such that
  for $s \in T$, $\pi'_v$ shifts the ranking of $u^*$
  to $n$ while maintaining the relatively rankings of all other
 elements in $\pi_s$, and $\pi'_v = \pi_v$ $\forall v\not\in T$.
Then, as the ranking of $u^*$ is more than $|S|$ in every $\pi'_s$ for
$s\in S$, we conclude that $S$ is not a \B3CT self-determined
community in $(V,\Pi')$, contradicting  the assumption that
 $S$ is stable under $\delta$-perturbations.
Now let $v^* = \mbox{argmax}\{\phi_S^\Pi(v): v\in V-S\}$, and let $\beta =
\phi_S^\Pi(v^*)/|S|$.
We can similarly show that
  if $S$ is stable under $\delta$-perturbations,
  then $\beta < \alpha - \delta$.

The second direction of the proposition is straightforward.
\end{proof}

Thus, the main result of \cite{B3CT} can be restated as:
   there are at most $n^{O(1/\delta)}$ \B3CT
   communities that are stable under $\delta$-perturbations.
  We further refine the stability
  studies of community functions by introducing the notion
  of membership-preserving perturbation:

\begin{definition}
Let $(V,\Pi)$ be a preference network, and let
${\emptyset\neq} S\subset V$.  A preference profile $\Pi'$
on $V$ is a {\em membership-preserving perturbation of
  $\Pi$ with respect to $S$} if  $\forall s\in S$,
$\pi_s(S)=\pi'_s(S)$.
\end{definition}

Note that the preference profile $\Pi'$ considered in
both Axiom {\sf CRNM} and {\sf CRM} are special cases of membership-preserving perturbations; in Axiom {\sf CRNM},
$\Pi_S$ and $\Pi'_S$  {\em agree on}
$S$ (i.e., for all $s, u \in S$, $\pi_s(u) = \pi'_s(u)$,
  implying in particular $\pi_s(S)=\pi'_s(S)$),
  and in Axiom {\sf CRM}, $\Pi_S$ and $\Pi'_S$  {\em agree on}
$V-S$ (i.e., for all $s\in S, v \in V-S$, $\pi_s(v) = \pi'_s(v)$,
implying $\pi_s(V-S)=\pi'_s(V-S)$) and hence also
$\pi_s(S)=\pi'_s(S)$.

\begin{theorem}
For any preference network $A = (V,\Pi)$,
  the number of \B3CT communities that are stable
  under membership-preserving, $\delta$-perturbations of $\Pi$
  is polynomial in $n^{O(1/\delta)}$.
\end{theorem}
\begin{proof}
It will be sufficient to show that if a \B3CT community
  $S$ is stable under membership-preserving,
   $\delta$-perturbations of $\Pi$,
   then
   either
\begin{enumerate}
\item $S$ is an $(\alpha,\alpha-\delta)$-\B3CT
  community for some $\alpha>\delta$, or
\item  $\exists s\in S$ such that $\pi_s(S)= [1:|S|]$.
\end{enumerate}
 Indeed, in the first case, there are at most $n^{O(1/\delta)}$
many $(\alpha,\alpha-\delta)$-\B3CT
  communities by  \cite{B3CT}, and in the second case,
we have that all communities are of the form
  $S=\pi_v^{-1}([1:k])$ for some $s\in V$ and $k\in [n]$,
  showing that
  there are at most $n^2$ such communities.

To establish the above statement,
  letting $u^* = \mbox{argmin}\{\phi_S^\Pi(u): u\in S\}$
  and $\alpha = \phi_S^\Pi(u^*)/|S|$,
  we now prove that if $S$ is not an
  $(\alpha,\alpha-\delta)$-\B3CT  community,
  then there must be $s\in S$, $\pi_s[1:|S|] = S$.
The assumption that $S$ is not an
$(\alpha,\alpha-\delta)$-\B3CT  community
implies  that
$\phi_S^\Pi(v^*)/|S| \geq \alpha -\delta$
  where $v^* = \mbox{argmax}\{\phi_S^\Pi(v): v\in V-S\}$.
Let $T =\{s\in S : \pi_s(v^*) \leq |S|\}$.
Then, $|T| = \phi_S^\Pi(v^*) \geq (\alpha - \delta)|S|$.
Since $S$ is a \B3CT community, we know that $|T| < \alpha |S|$.

Using these conditions, we now define a perturbed preference profile.
Key to our construction is the following observation:
For each $s\in S-T$, if $\pi_s(S) \neq [1:|S|]$,
   then  $(V-S)\cap \pi_s[1:|S|] \neq \emptyset$.
Thus, there exists  $\pi'_s$ that agrees with $\pi_s$ on $S$ and
    $\pi'_{s}(v^*)\leq |S|$ -- we can simply swap $v^*$ with any
    element in $(V-S)\cap \pi_s[1:|S|]$.
Thus, either there exists $s\in S-T$ such that $\pi_s(S) = [1:|S|]$
 (which implies Case 2 above), or
  $\pi_s(S)\neq [1:|S|] ,\forall s \in S-T$.
The latter implies that we can find
     a set $\tilde S\subset S-T$ of size
  $\alpha |S| - |T| \leq \delta |S|$  and a
membership-preserving, $\delta$-perturbations $\Pi'$ of
$\Pi$ such that
    $\pi'_{s}(v^*)\leq |S|$ for   all $s\in T\cup\tilde S$, implying that
$S$ is not a \B3CT community in
$(V,\Pi')$.  This contradicts the assumption that
$S$ is stable under membership-preserving,
$\delta$-perturbations of $\Pi$.
\end{proof}

\subsection{Stable Fixed-Points of Social Choice }\label{sec:StableFixedPoints}

We can also strengthen the concept of fixed points in our social
   choice based community framework.
Particularly, we
  measure the stability of a community
  defined by a fixed-point rule according to some variation of
  the following definition.

\begin{definition}{\sc ($\delta$-Strong Fixed Points)}\label{DeltaSocialChoiceVoting}
Let $A = (V,\Pi)$ be a preference netowrk,
$F: \LV^* \rightarrow \bLV$ be a
  preference aggregation function, and $\delta \in [0:1]$ be a coherence
  parameter.
Then, $S\in \calC_F((V,\Pi))$  is {\em $\delta$-strong} if
for $\forall T\subseteq S$ such that $|T| \geq
  (1-\delta)\cdot |S|$,
  \[
  u\succ_{ F(\Pi_T)} v\qquad\forall u\in S, v\in  V-S.
  \]
\end{definition}

Our goal is to
  understand the
  influence of a preference aggregation function $F$ and
  the stability parameter $\delta$ ($0 \leq \delta \leq 1$) on the
   structure of the $\delta$-strong
   $F$-self-determined communities.

 Before discussing this further, we point out some
subtleties that arise when applying
Definition
\ref{DeltaSocialChoiceVoting}
to general aggregation functions.
We illustrate this subtlety using weighted
  fixed-point rules, and,
 in particular, by comparing the community rule defined
  by the \B3CT voting function  to that defined by the Borda count.

Recall that in Definition \ref{def:FPR},
  for preference networks with $n$ elements, a preference aggregation function
  is determined by a sequence of weighting vectors
  $W = (w^1,w^2,\dots)$
  where $w^k\in \mathbb{R}^n$, denotes the weights for the aggregation of $k$ preferences.
While this weight vector
  is independent of $k$ for the Borda count, it in general can
be  different for each $k$, and indeed does depend on $k$ for
 \B3CT voting.  Concretely, for the
   Borda count, every voting member gives  scores
 $n,n-1,\dots,1$ to the members of $V$,  while in  \B3CT voting,
 it gives a score of $1$ to the
  first $k$ in her preference list, making her scores
 dependent on the total number of voters, $k$.
Thus when defining self-determined communities with the
  Borda count, one does not need to first
 anticipate the community size before aggregating the preferences
 of its members, but when defining self-determined communities with
  \B3CT voting, the weight assigned to an element by a preference
  depends on the size of the subset under consideration.

In this regard, when measuring the stability of a community $S$,
 Definition \ref{DeltaSocialChoiceVoting} uses
  the same weighting vector to evaluate $F(\Pi_T)$ and $F(\Pi_S)$
  for the Borda count based community rule,
while it uses different weighting vectors to
  evaluate $F(\Pi_T)$ and $F(\Pi_S)$
  for the \B3CT community rule, and these weighting vectors
  depend on $|T|$.
Thus, the former application of Definition \ref{DeltaSocialChoiceVoting}
  appears more natural than the latter application.

As a result, we will use the following variation of
  Definition \ref{DeltaSocialChoiceVoting} to measure the strength
  of a \B3CT community.

\begin{definition}[$\delta$-Strong \B3CT Communities]\label{def:StrongB3CT}
For each $T\subseteq V$ and $i\in V$, let $\phi_{T,k}^\Pi(i)$
denote the number of votes that member $i$ would receive if each
member $s\in T$ were casting a vote for each of its $k$ most
preferred members according to its preference $\pi_s$.

For $\delta \in [0: 1]$, a non-empty
 set $S\subseteq V$ is a {\em $\delta$-strong \B3CT community} in
  $A = (V,\Pi)$
  if $\forall u\in S, v \in V -S$ and $T \subseteq S$ such that
  $|T| \geq (1-\delta)\cdot |S|$,
\[\phi_{T,|S|}^\Pi(u) > \phi_{T,|S|}^\Pi(v) \qquad\forall u\in S, v\in  V-S.
\]
\end{definition}

\begin{proposition}\label{prop:strong}
If $S$ is a $\delta$-strong \B3CT community
of $A = (V,\Pi)$, then  $\exists \alpha \geq \delta$ such that $S$
 is an $(\alpha,\alpha-\delta)$-\B3CT community.
\end{proposition}
\begin{proof}
Let $u^* = \mbox{argmin}\{\phi_S^\Pi(u): u\in S\}$, and let $\alpha =
\phi_S^\Pi(u^*)/|S|$ and let
$v^* = \mbox{argmax}\{\phi_S^\Pi(v): v\in V-S\}$, and let $\beta =
\phi_S^\Pi(v^*)/|S|$.
We now prove that $\alpha - \beta > \delta$.

The pair $u^*$ and $v^*$ partitions $S$ into four subsets.
\begin{eqnarray*}
S_0 & = & \{s \in S: (\pi_s (u^*)\not\in [1:|S|])\mbox{ and }(\pi_s (v^*)\not\in
[1:|S|])\}\\
S_1 & = & \{s \in S: (\pi_s (u^*)\not\in [1:|S|])\mbox{ and }(\pi_s (v^*)\in
[1:|S|])\}\\
S_2 & = & \{s \in S: (\pi_s(u^*)\in[1:|S|])\mbox{ and }(\pi_s (v^*)\not\in [1:|S|])\}\\
S_3 & = & \{s \in S: (\pi_s (u^*)\in [1:|S|])\mbox{ and }(\pi_s (v^*)\in
[1:|S|])\}
\end{eqnarray*}
Then
\begin{eqnarray*}
|S_0| + |S_1| + |S_2| + |S_3| & = & |S|,\label{con1}
\\
|S_1| +  |S_3|  & = &\beta\cdot |S|,
\\
|S_2| + |S_3| & = &\alpha\cdot |S|,
\\
|S_0| + 2|S_1| + |S_3| & < &(1-\delta)\cdot|S|,
 \end{eqnarray*}
 where
the last inequality follows from the assumption that $S$ is a
  $\delta$-strong \B3CT-self-determined community.

   To see this, we
  first note that $|S_1| <|S_2|$ due to the fact that $\beta < \alpha$.
Define $T$ to be the union of $S_0$, $S_1$, $S_3$, and
  $\tilde S_2$, where $\tilde S_2\subset S_2$
is an arbitrary subset of size $|S_1|$.  Assume
 by contradiction that
$|T|\geq (1-\delta)|S|$.  Since $S$ is  a
  $\delta$-strong \B3CT-self-determined community,
  this would imply that $\phi_{T,|S|}(u^*)>\phi_{T,|S|}(v^*)$, i.e.
  \[
 0< \sum_{s\in T}\Bigl(1_{\pi_s(u^*)\leq |S|}-1_{\pi_s(v^*)\leq |S|}\Bigr).
  \]
But the right hand side is
  equal to $|\tilde S_2|-|S_1|=0$, leading to a contradiction.
Therefore, $|T| < (1-\delta)\cdot |S|$, as claimed.

Subtracting the fourth of the above equations from the first, we obtain
$(|S_0| + |S_1| + |S_2| + |S_3|)-(|S_0| + 2|S_1| + |S_3|)
 =
(|S_2| + |S_3|)-(|S_1| + |S_3|) > |S| - (1-\delta)\cdot|S| = \delta\cdot |S|
$.
Thus, by the second and third equation, we have
$\alpha\cdot|S| - \beta\cdot|S|>\delta\cdot |S|.$
\end{proof}

\begin{proposition}
For any $\delta \in (0,1)$,
  the number of $\delta$-strong
  \B3CT communities
  in any preference network is $n^{O(1/\delta)}$.
\end{proposition}
\begin{proof}
This follows from the main result of \cite{B3CT} and Proposition
\ref{prop:strong} above.
\end{proof}

\subsection{Stable Harmonious Communities}\label{subsec:propHarmonious}

Applying the stability notions of Sections \ref{sec:perturbation}
and  \ref{sec:StableFixedPoints},
we define two types of stable harmonious communities.
Before doing so, we recall the definition of harmonious communities,  Definition~\ref{def:harmonious}, and the definition of
$\lambda$-harmonious communities from Rule~\ref{rule:lambda-harm}.

\begin{definition}[Stable Harmonious Communities]\label{def:StableHarmonious}
 For $\delta \in [0: 1/2]$, a non-empty subset $S$ is a
{\em $\delta$-stable harmonious community} in $A  = (V,\Pi)$
 if $S$ is $(\delta+1/2)$-harmonious, i.e., if
$\forall u\in S, v \in V -S$, at least $(1/2+\delta)$-fraction
  of  $\{\pi_s\ : \ s\in S\}$  prefer $u$ over $v$.
For $\delta \in [0: 1]$, $S$ is a
{\em $\delta$-strong harmonious community} in $A$
if $\forall u\in S, v \in V -S$ and $T \subseteq S$ such that
  $|T| \geq (1-\delta)\cdot |S|$, the
   majority of  $\{\pi_s\ : \ s\in T\}$  prefer $u$ over $v$.
\end{definition}

Note that a $\delta$-stable harmonious community
is not quite the same as a harmonious community stable
under $\delta$-perturbations as defined in Section \ref{sec:perturbation}.
Instead, we have that
a $\delta$-stable harmonious community is  a harmonious
community that is stable under any $\delta'$-perturbations as long as
$\delta'<\delta/2$, and that conversely, a harmonious community that is stable
under $\delta$ perturbation is a $\delta$-stable harmonious community.
By contrast, the definition of $\delta$-strong harmonious communities
maps exactly to the definition given in Section~\ref{sec:StableFixedPoints}.

\begin{proposition}\label{prop:strongH}
If $S$ is a $\delta$-strong harmonous community,
then $S$ is a $\delta/2$-stable harmonious community.
\end{proposition}
\begin{proof}
For each pair $u\in S,v\in V-S$, let $f(u,v) = |\{s\in S: u\succ_{\pi_s}
v\}|- |\{s\in S: u\prec_{\pi_s} v\}|$ be the {\em preference gap} between
$u$ and $v$ with respect to $S$.
Suppose $(u^*,v^*) = \mbox{argmin}\{f(u,v) :u\in S,v\in V-S\}$.
We now show that if $S$ is a $\delta$-strong harmonous community
 of $A$, then $f(u^*,v^*) > \delta\cdot|S|$.
The pair $u^*$ and $v^*$ partitions $S$ into two subsets.
$S_\succ = \{s\in S: u^*\succ_{\pi_s} v^*\}$ and
$S_\prec = \{s\in S: u^*\prec_{\pi_s} v^*\}$.
We have $|S_\succ| + |S_\prec|   = |S|$ and $|S_\succ| > |S_\prec|$.
Let $T$ be the union of $S_\prec$ and $|S_\prec|$ arbitrary members of $S_\succ$.
Since members of $T$ are indifferent about $u^*$ and $v^*$,
  we have $|T|= 2|S_\prec| \leq (1-\delta)\cdot |S|$.
Thus
$
f(u^*,v^*)  =  |S_\succ| - |S_\prec| = (|S_\succ|+|S_\prec|) - 2|S_\prec|
  \geq  |S| -
(1-\delta)\cdot |S| = \delta\cdot |S|.
$
Thus, $|S_\succ| > (1/2+\delta/2)\cdot |S|$, and at least
$(1/2+\delta/2)$-fraction of $\Pi_S$ prefer $u^*$ over $v^*$.
\end{proof}

With a simple probabilistic argument, we can
 bound the number of
$\delta$-stable harmonious communities in any preference networks.

\begin{theorem}
$\forall \delta \leq 1/2$, the number of $\delta$-stable harmonious communities in
 any preference network is $n^{12\log n/\delta^2}$.
\end{theorem}
\begin{proof}
Let $S$ be a $\delta$-stable harmonious communities.
For any multi-set $T\subseteq S$,
 we say $T$ {\em identifies} $S$
 if for all $u\in S$ and $v\in V-S$, the majority of $T$
 prefer $u$ to $v$.
Note that such a $T$  determines $S$ once the size of $S$ is set.  To see this,
note that the condition implies that $u\succ_{F(\Pi_T)} v$ for all $(u,v)\in S\times(V-S)$,
which in turn implies that $S$ is of the form $V_1\cup\dots\cup V_i$ where $(V_1,V_2,\dots)$
are the components of the ordered partition $F(\Pi_T)$, ordered in such a way
that $V_1\succ_{F(\Pi_T)}V_2, ...$ (see Proposition~\ref{prop:harmonious-aggregation} and its proof).
Thus once $F(\Pi_T)$ and the size of $S$ are fixed,
$S$ is uniquely determined.

We now show that  $\exists T\subset V$ of size $12\log n/\delta^2$ that identifies $S$.
To this end, we consider a sample $T\subset S$
  of $k = 12\log n/\delta^2$ randomly chosen elements (with replacements).
We analyze the probability that $T$ identifies $S$.
Let $T = \{t_1,...,t_k\}$, and for each $u\in S$ and $v\in V-S$,
  let $x^{(u,v)}_i = [u\succ_{\pi_{t_i}} v]$,
  where $[B]$ denotes the indicator varable of an event $B$.
Then $T$ identifies $S$ iff
  $\sum_{i=1}^k  x^{(u,v)}_i > k/2, \forall u\in S, v\in V -  S$.
We now focus on a particular $(u,v)$ pair and bound
  $\prob{}{\sum_{i=1}^k  x^{(u,v)}_i \leq k/2}.$
We first note that
$$\expec{}{\sum_{i=1}^k  x^{(u,v)}_i }
= \sum_{i=1}^k \expec{}{x^{(u,v)}_i}
\geq \left(\frac{1}{2}+\delta\right)\cdot k.$$
By a standard use of the Chernoff-Hoeffding bound
\begin{eqnarray*}
& & \hspace{-0.5in}\prob{}{\sum_{i=1}^k  x^{(u,v)}_i  \leq k/2}
\\ & \leq &
\prob{}{ \sum_{i=1}^k  x^{(u,v)}_i\leq (1+2\delta)^{-1}  \expec{}{\sum_{i=1}^k  x^{(u,v)}_i }}
\\ & \leq &
 \prob{}{ \sum_{i=1}^k  x^{(u,v)}_i\leq (1-\delta) \expec{}{\sum_{i=1}^k  x^{(u,v)}_i }}
\\
&\leq & e^{ -{\frac{\delta^2}2} (1/2+\delta)k}
 \leq e^{-\frac{\delta^2}{4}k}
 \leq \frac{1}{n^3}, \end{eqnarray*}
 where we used that
$(1+2\delta)^{-1}= 1-2\delta(1+2\delta)^{-1}\leq 1-\delta$ in the third step.

If $T$ fails to identify $S$, then there exists
  $(u\in S,v\in V-S)$ such that
  $\sum_{i=1}^k  x^{(u,v)}_i  \leq k/2$.
As there are at most $|S||V-S|\leq n^2$ such
   $(u,v)$ pairs to consider,
 by the union bound,
\begin{eqnarray*}
\prob{}{\mbox{$T$ identifies $S$}} & \geq &
   1- \sum_{u\in S,v\in V-S}\prob{}{\sum_{i=1}^k  x^{(u,v)}_i  \leq
   k/2}\\
& > & 1-1/n{ >0}.
\end{eqnarray*}
Thus, if $S$ is a $\delta$-stable harmonious communities,
 then  there exists a multi-set $T\subset V$
of size $12\log n/\delta^2$ that identifies $S$.
We can thus enumerate all $\delta$-stable harmonious communities
  by enumerating all $(T,t)$ pairs, where $T$ ranges from
  all multi-subsets of
  $V$ of size $12\log n/\delta^2$ and $t \in [1:n]$ and check if
$T$ can identify a set of size $t$.
\end{proof}

\section{Remarks}\label{Sec:Remarks}
While the
  results of this paper are conceptual and are built on the
  abstract framework of preference networks,
  we hope this study is a significant step towards developing a
  rigorous theory of community formation in social and information
  networks. In particular, we hope this will be used to inform and choose
  among other approaches to community identification
  which have been developed.
Below we discuss a few short-term research directions
  that may help to expand our understanding in order to make more effective
  connection with community identification in
   networks that arise in practice.

\subsubsection*{Preferences Models}

We have based our community formation theory on the {\em ordinal}
  concept of utilities used in social choice and modern economic
  theory \cite{ArrowBook}.
The resulting preference network framework, like that in the classic
   studies of voting \cite{ArrowBook} and stable marriage
 \cite{stableMarriage},
  enables our axiomatic approach to focus on the conceptual question
  of network communities rather than the more practical question of
  community formation in an observed social network.
To better connect with the real-world community identification problem, we
  need to loosen both the assumption of strict ranking and
  the assumption of complete preference information.


With simple modifications to our axioms,
  we can extend our entire theory to a
  preference network  $A= (V,\Pi)$ that allows {\em indifferences},
  i.e., $\Pi$ is given by $n$ ordered partitions $\{\pi_1,...,\pi_n\}:
  \pi_i \in \bLV$.
This extension enables us to partially expand our results to
  affinity networks.
Recall an affinity network $A = (V,W)$ is given by $n$ vectors
  $W = \{w_1,...,w_n\}$, where $w_i$ is an $n$-place non-negative
  vectors.
We can extract an ordinal preference $\pi_i \in \bLV$ from the cardinal
  affinities by sorting entries in $w_i$ -- elements with the same
  weight are assigned to the same partition.

Although this conversion may lose some valuable affinity information
  encoded in the numerical values, it offers a path for us to apply our
  community theory -- even in its current form -- to network analysis.
For example, as suggested in \cite{B3CT},
  given a social network $G = (V,E)$, we can first define
  an affinity network $A = (V,W)$ where $w_i$ is the personalized
  PageRank vector of vertex $i$, and then obtain
  an preference network $(V,\Pi)$
  where $\pi\in \bLV$ ranks vertices in $V$ by $i$'s PageRank
  contributions \cite{PageRankContribution} to them.

Theoretically, we would like to extend our work
  to preference networks with partially ordered preferences
  as a concrete step to understand community formation in
  networks with incomplete or incomparable preferences.
Like our current study, we believe that the existing literature in social
  choice -- e.g., \cite{PartialOrder} -- will be valuable to our
  understanding.
We expect that an axiomatic community approach
   to preference networks with partially ordered preferences,
   together with an axiomatization theory of
  personalized ranking in a network, may offer us
  new understanding of how to address the two basic mathematical problems
  --  extension of individual affinities/preferences to community
  coherence and inference of missing links --
  for studying communities in a social and
  information network.
As this part of community theory becomes sufficiently well developed,
   well-designed experiments with real-world social networks
   will be necessary to further enhance this theoretical framework.

\subsubsection*{Structures, Algorithms, and Complexity}

Our taxonomy theorem provides the basic structure of communities in a
preference network, while the {\sf coNP}-Completeness result illustrates
the algorithmic challenges for community identification in addition to
community enumeration. On the other hand, our analysis of the harmonious rule and the work of
\cite{B3CT} seem to suggest some efficient notion of communities can
be defined.

However,  it remains an open question if there exists a natural and constructive
community rule that simultaneously (i) satisfies all axioms, 
(ii) allows overlapping communities, and (iii) has stable communities which are
polynomial-time samplable and enumerable.

\section*{Acknowledgements}
We thank Nina Balcan, Mark Braverman, and  Madhu Sudan 
  for all the insightful discussions.
The complexity proof of Proposition \ref{prop:H-complexity} 
  is due to Madhu.
\bibliographystyle{abbrv}
\bibliography{community}
\end{document}